\documentclass[onecolumn,draftclsnofoot,nofonttune]{IEEEtran}

\usepackage{amsmath}
\usepackage{amssymb}
\usepackage{amsfonts}
\usepackage{amsthm}

\usepackage[textwidth=20mm,disable]{todonotes}
\usepackage[utf8]{inputenc}
\usepackage[T1]{fontenc}
\usepackage{txfonts}
\let\mathbb=\varmathbb

\usepackage[%
cal=euler,
frak=euler
]
{mathalfa}
\usepackage[mathcal]{euscript}

\usepackage{subfigure}
\usepackage{tikz}
\usepackage{xcolor}
\definecolor{halfgray}{gray}{0.55}
\definecolor{OliveGreen}{rgb}{0,.35,0}
\definecolor{webbrown}{rgb}{.6,0,0}
\definecolor{BrightViolet}{rgb}{0.5,0.2,0.8}
\definecolor{Maroon}{cmyk}{0, 0.87, 0.68, 0.32}
\definecolor{RoyalBlue}{cmyk}{1, 0.50, 0, 0.25}
\definecolor{Black}{cmyk}{0, 0, 0, 0}

\usepackage{acronym}
\usepackage{balance}
\usepackage{cancel}
\usepackage{latexsym}
\usepackage{paralist}
\usepackage{wasysym}
\usepackage{xspace}

\usepackage[numbers,sort&compress]{natbib}

\usepackage{hyperref}
\hypersetup{
colorlinks=true,
linktocpage=true,
pdfstartview=FitH,
breaklinks=true,
pdfpagemode=UseNone,
pageanchor=true,
pdfpagemode=UseOutlines,
plainpages=false,
bookmarksnumbered,
bookmarksopen=false,
bookmarksopenlevel=0,
hypertexnames=true,
pdfhighlight=/O,
urlcolor=Maroon, linkcolor=RoyalBlue, citecolor=OliveGreen, 
pdftitle={},
pdfauthor={},
pdfsubject={},
pdfkeywords={},
pdfcreator={pdfLaTeX},
pdfproducer={LaTeX with hyperref}
}

\newcommand{\C}{\mathbb{C}}
\newcommand{\R}{\mathbb{R}}

\DeclareMathOperator{\ex}{\mathbb{E}}

\newcommand{\grad}{\nabla}

\DeclareMathOperator{\hess}{Hess}

\DeclareMathOperator*{\argmax}{\arg\max}

\DeclareMathOperator{\supp}{supp}

\newcommand{\from}{\colon}

\newcommand{\mgeq}{\succcurlyeq}

\newcommand{\pd}{\partial}

\newcommand{\varsimplex}{\mathcal{D}}

\newcommand{\dis}{\displaystyle}
\newcommand{\txs}{\textstyle}

\newcommand{\insum}{\sum\nolimits}

\usepackage{soul}
\sethlcolor{pink}

\usepackage{algorithm2e_extension}
\usepackage{algorithm}
\SetKwBlock{Repeat}{Repeat}{until}

\newtheorem{theorem}{Theorem}

\newtheorem*{corollary*}{Corollary}

\newtheorem{proposition}{Proposition}

\theoremstyle{definition}

\newtheorem*{definition*}{Definition}

\theoremstyle{remark}
\newtheorem{remark}{Remark}
\newtheorem*{remark*}{Remark}

\newcommand{\bp}{\mathbf{p}}

\newcommand{\bv}{\mathbf{v}}
\newcommand{\bw}{\mathbf{w}}
\newcommand{\bx}{\mathbf{x}}

\newcommand{\bz}{\mathbf{z}}

\newcommand{\eq}{p^{\ast}}
\newcommand{\eqvec}{\bp^{\ast}}
\newcommand{\eqset}{\varsimplex^{\ast}}

\newcommand{\sinr}{\mathsf{sinr}}

\newcommand{\breg}{D_{h}}

\newcommand{\gibbs}{G}
\newcommand{\gibbsvec}{\mathbf{\gibbs}}
\newcommand{\pay}{u}
\newcommand{\payv}{v}
\newcommand{\payvec}{\mathbf{\payv}}
\newcommand{\score}{y}
\newcommand{\scorevec}{\mathbf{\score}}

\newcommand{\noisedev}{\sigma}
\newcommand{\noisevar}{\noisedev^{2}}
\newcommand{\LP}{\mathrm{LP}}
\newcommand{\VP}{\mathrm{VP}}

\newcommand{\PU}{\mathrm{PU}}
\newcommand{\eql}{\mathsf{EQL}}
\newcommand{\VI}{\Psi}

\newcommand{\step}{\gamma}

\newcommand{\abs}[1]{\left\lvert #1 \right\rvert}

\newcommand{\norm}[1]{\left\| #1 \right\|}

\newcommand{\braket}[2]{\left\langle #1 \middle\vert  #2 \right\rangle}

\newcommand{\set}{\mathcal{S}}
\newcommand{\subs}{\mathcal{S}}
\newcommand{\sub}{s}

\newcommand{\rate}{R}
\newcommand{\cost}{C}
\newcommand{\price}{\pi}
\newcommand{\pot}{V}
\newcommand{\strat}{\mathcal{X}}
\newcommand{\game}{\mathfrak{G}}
\newcommand{\play}{\mathcal{K}}
\newcommand{\intf}{I^{\mathrm{max}}}

\newcommand{\PM}[1]{\todo[inline,color=OliveGreen!40,author=\textbf{\small\scshape PM says}]{\footnotesize #1\\}}

\newcommand{\SD}[1]{\todo[inline,color=RoyalBlue!40,author=\textbf{\small\scshape SD says}]{\footnotesize #1\\}}

\begin{document}


\title{Cost-Efficient Throughput Maximization in Multi-Carrier Cognitive Radio Systems}

\author{%
Salvatore D'Oro,%
~\IEEEmembership{Student Member,~IEEE},
Panayotis Mertikopoulos,%
~\IEEEmembership{Member,~IEEE},
Aris L. Moustakas,%
~\IEEEmembership{Senior Member,~IEEE},
Sergio Palazzo,%
~\IEEEmembership{Senior Member,~IEEE}
\thanks{%
This research was supported in part by the European Commission in the framework of the FP7 Network of Excellence in Wireless COMmunications NEWCOM\# (contract no. 318306),
and by the French National Research Agency project
NETLEARN (ANR\textendash 13\textendash INFR\textendash 004).
Part of this work was presented at WiOpt 2014: the 12th International Symposium on Modelling and Optimization in Mobile, Ad Hoc, and Wireless Networks, Hammamet, Tunisia, May~2014.}
\thanks{%
S. D'Oro and S. Palazzo are with the CNIT Research Unit, University of Catania, Italy;
P.~Mertikopoulos is with the French National Center for Scientific Research (CNRS) and the Laboratoire d'Informatique de Grenoble, Grenoble, France;
A. L. Moustakas is with the Department of Physics, University of Athens, the Institute of Accelerating Systems and Applications (IASA), Athens, Greece, and the École Supérieure d'Électricité (Supélec), Gif-sur-Yvette, France, supported by the Digiteo Senior Chair "ASAPGONE".}
}


\maketitle

\newacro{5G}{5th generation}
\newacro{ICT}{information and communications technology}
\newacro{CR}{cognitive radio}
\newacro{CCI}{co-channel interference}
\newacro{IT}{interference temperature}
\newacro{MC}{multi-carrier}
\newacro{MAC}{multiple access channel}
\newacro{MIMO}{multiple-input and multiple-output}
\newacro{MUI}{multi-user interference-plus-noise}
\newacro{NE}{Nash equilibrium}
\newacroplural{NE}[NE]{Nash equilibria}
\newacro{DSM}{dynamic spectrum management}
\newacro{PMAC}{parallel multiple access channel}
\newacro{CSI}{channel state information}
\newacro{CSIT}{channel state information at the transmitter}
\newacro{CDMA}{code division multiple access}
\newacro{DSL}{digital subscriber line}
\newacro{SIC}{successive interference cancellation}
\newacro{SUD}{single user decoding}
\newacro{SINR}{signal-to-interference-and-noise ratio}
\newacro{KKT}{Ka\-rush--Kuhn--Tuc\-ker}
\newacro{IWF}{iterative water-filling}
\newacro{SWF}{simultaneous water-filling}
\newacro{PU}{primary user}
\newacro{SU}{secondary user}
\newacro{iid}[i.i.d.]{independent and identically distributed}
\newacro{XL}{exponential learning}
\newacro{FCC}{Federal Communications Commission}
\newacro{NTIA}{National Telecommunications and Information Administration}
\newacro{GAO}{General Accounting Office}
\newacro{QoS}{quality of service}
\newacro{OFDM}{orthogonal frequency division multiplexing}\acused{OFDM}
\newacro{ILP}[I-LP]{interference linear pricing}
\newacro{IPLP}[I-PLP]{interference proportional linear pricing}
\newacro{IPEP}[I-PEP]{interference proportional exponential pricing}
\newacro{LP}{linear pricing}
\newacro{VP}{violation pricing}
\newacro{NLP}[EP]{exponential pricing}
\newacro{EQL}{equilibration level}
\newacro{KL}{Kullback\textendash Leibler}
\newacro{APT}{asymptotic pseudotrajectory}
\newacro{STC}{search-then-converge}
\newacro{TDD}{time-division duplexing}

\begin{abstract}
Cognitive radio (CR) systems allow opportunistic, \acp{SU} to access portions of the spectrum that are unused by the network's licensed \acp{PU}, provided that the induced interference does not compromise the \acs{PU}' performance guarantees.
To account for interference constraints of this type, we consider a flexible spectrum access pricing scheme that charges \acp{SU} based on the interference that they cause to the system's \acp{PU} (individually, globally, or both), and we examine how \acp{SU} can maximize their achievable transmission rate in this setting.
We show that the resulting non-cooperative game admits a unique Nash equilibrium under very mild assumptions on the pricing mechanism employed by the network operator, and under both static and ergodic (fast-fading) channel conditions.
In addition, we derive a dynamic power allocation policy that converges to equilibrium within a few iterations (even for large numbers of users), and which relies only on local \ac{SINR} measurements;
importantly, the proposed algorithm retains its convergence properties even in the ergodic channel regime, despite the inherent stochasticity thereof.
Our theoretical analysis is complemented by extensive numerical simulations which illustrate the performance and scalability properties of the proposed pricing scheme under realistic network conditions.
\end{abstract}


\acresetall

\section{Introduction}
\label{sec:intro}

Greatly raising the bar from previous generation upgrades, current design specifications for \ac{5G} wireless systems target a massive increase in network capacity, fiber-like connection speeds (well into the Gb/s range), and an immersive overall user experience with zero effective latency and response times \cite{ABC+14,Qualcomm}.
As such, the \acs{ICT} industry is faced with a formidable challenge:
these ambitious design goals require the deployment of new wireless interfaces at an unprecedented scale, but the necessary overhaul is limited by the inherent constraints of upgrading an entrenched (and often ageing) wireless infrastructure.

Chief among these concerns is the projected spectrum crunch:
if not properly managed, the existing radio spectrum will not be able to accommodate the soaring demand for wireless broadband and the ever-growing volume of data traffic \cite{FCC02}.
To make matters worse, studies by the US \ac{FCC} and the \ac{NTIA} have shown that this vital commodity is effectively squandered through underutilization and inefficient use:
for instance, only $15\%$ to $85\%$ of the licensed radio spectrum is used on average, leaving ample spectral voids that could be exploited via efficient spectrum management techniques \cite{FCC02,GAO04}.
Accordingly, in this often unregulated context, the emerging paradigm of \ac{CR} has attracted considerable interest as a promising way out of the spectrum gridlock \cite{MM99,ZS07,Hay05,GJMS09}.

At its most basic level, \acl{CR} comprises a two-level hierarchy between wireless users induced by spectrum licencing:
the network's licensed, \acp{PU} have purchased spectrum rights from the network operator (often in the form of contractual \ac{QoS} guarantees), but they allow unlicenced \acp{SU} to access the spectrum provided that the induced \ac{CCI} remains below a certain threshold \cite{MM99,Hay05}.
Put differently, by sensing the wireless medium, the network's cognitive \acp{SU} essentially free-ride on the \acp{PU}' licensed spectrum and they try to communicate under the constraints imposed by the \acp{PU} (though, of course, without any \ac{QoS} guarantees).
Thus, by opening up the unused part of the spectrum to opportunistic user access, overall spectrum utilization is increased without needing to deploy more (and more expensive) wireless interfaces \cite{ZS07,SSSBN08}.

Of course, given the non-cooperative nature of this opportunistic framework, throughput optimization in \ac{CR} environments calls for flexible and decentralized optimization policies with minimal information exchange between \acp{SU}, \acp{PU}, and access points/base stations.
In particular, a major challenge involves safeguarding the performance guarantees that the network's licensed \aclp{PU} have already paid for:
if \aclp{SU} are allowed to transmit without some power/interference control mechanism in place, then the \aclp{PU}' \ac{QoS} requirements may be violated, thus invalidating the fundamental operational premise of \ac{CR} systems.
To that end, the authors of \cite{ABSA02,SMG02,SBP06,DMMP14} investigated the role of pricing as an effective mechanism to control interference and they provided an energy/cost-efficient formulation of the problem where users seek to maximize their transmission rate while keeping their transmit power in check.
To reach a stable equilibrium state in this setting, several distributed approaches have been proposed, based chiefly on reaction functions \cite{ABSA02}, Gauss-Seidel and Jacobi update algorithms \cite{SBP06}, or learning methods \cite{MBML12,DMMP14};
however, these works do not distinguish between licenced and unlicensed users, so their results do not immediately apply to \ac{CR} networks.

In \ac{CR} systems, \ac{PU} requirements are often treated as \ac{IT} \cite{FCC03} constraints that are coupled across the network's \acp{SU}, and the theoretical analysis of the resulting system
aims to characterize the network's optimum/unilaterally stable equilibrium states and to provide the means to converge to such states \cite{NH08,WWL2010game,PSPF10,YSSP13,XZ14}.
These constraints are then enforced indirectly via exogenous pricing mechanisms that charge \acp{SU} based on the aggregate interference that they cause to the network's \acp{PU} (and, of course, \acp{PU} are reimbursed commensurately).
In this context, the authors of \cite{NH08} introduced a spectrum-trading mechanism based on a market-equilibrium approach \cite{MWG95} and they provided an algorithm allowing \acp{SU} to estimate spectrum prices and adjust their spectrum demands accordingly.
More recently, to account for the \acp{PU}' maximum interference tolerance, the authors of \cite{PSPF10,YSSP13} introduced a game-theoretic formulation of \ac{CR} interference channels where \acp{SU} are charged proportionally to the aggregate interference caused;
then, using variational inequality methodologies, they derived sufficient low-interference conditions under which the resulting game admits a unique Nash equilibrium and they proposed a best-response algorithm that converges to this equilibrium state.
The case of inexact system information was considered in \cite{XZ14} where the authors formulated the problem as a (deterministic) robust optimization program which can be solved by Lagrangian dual decomposition methods.
A game-theoretic account of the impact of \ac{IT} constraints on system performance is also studied in \cite{WCPW07} where the authors derive cost-aware optimal power allocation policies by relaxing the problem's hard \ac{IT} constraints and incorporating an exponential cost in the \acp{SU}' utility functions;
in this context, the resulting power allocation game admits a unique equilibrium which is also Pareto efficient in the low-interference regime.
Finally, by exploiting the innate hierarchy between primary and secondary users, the authors of \cite{WJH14} provided a Stackelberg game formulation where the system's \ac{PU} acts as the leader and seeks to maximize the revenue generated by discriminatory spectrum access pricing mechanisms imposed to \acp{SU} (the game's followers).
\PM{Salvo, please input your own literature review here w.r.t. \cite{WJH14,WCPW07}. Sergio and Aris, your contributions are more than welcome, too!}
\SD{Done! However, I have not access to \cite{WJH14}, so my review is based on what I remember was studied in the paper. I will revise it as soon as I get access to the paper. Leave this comment as a reminder, please.}
\PM{I killed the last phrase on backward induction and noncavity to save space \textendash\ we can always add it back if needed.}

That being said,
the above works focus almost exclusively on wireless systems with static channel conditions where the benefits of interference control mechanisms are relatively easy to evaluate;
by contrast, very little is known in the case where the channels vary with time (e.g., due to user mobility).
In the presence of (fast) fading, channel gains are typically assumed to follow a stationary ergodic process, so the users' throughput and induced interference depend crucially on the channel statistics.
In this stochastic framework, the authors of \cite{MBML12} studied the problem of ergodic rate maximization in \ac{MC} systems and derived an efficient power allocation algorithm that allows users to attain the system's capacity;
however, no distinction was made between licensed and unlicensed users, so the results of \cite{MBML12} do not readily translate to a \ac{CR} setting.
More recently, \cite{MB14} provided an efficient online learning algorithm for unilateral rate optimization in dynamic \acl{MC} \ac{MIMO} \acl{CR} systems, but, again, without taking into account any \ac{IT} constraints imposed by the network's \aclp{PU}.

In this paper, we consider the problem of cost-efficient throughput maximization in \acl{MC} \acl{CR} networks where \acp{SU} are charged based on the interference that they cause to the network's \acp{PU} (either on an aggregate or a per-user basis).
Our system model is presented in Section \ref{sec:model} where we consider a general game-theoretic formulation that is flexible enough to account for both aggregate (flat-rate), temperature-based, and per-user pricing schemes.
In the case of static channels (Section \ref{sec:learning}), we show that the resulting game admits a unique Nash equilibrium almost surely, provided that the \acp{SU}' pricing schemes satisfy some fairly mild requirements (for instance, that a user's transmission cost increases with his radiated power).
On the other hand, in the case of fast-fading channels (which we study in Section \ref{sec:fading}), we show that the game under study admits a unique Nash equilibrium always, without any further caveats.

Moreover, extending the exponential learning techniques of \cite{MBML12}, we also derive a dynamic power allocation policy that converges to Nash equilibrium in a few iterations, even for large numbers of users and/or subcarriers per user.
In particular, the proposed algorithm has the following desirable attributes:
\begin{enumerate}
\item
\emph{Distributedness:}
user updates are based on local information and signal measurements.
\item
\emph{Statelessness:}
users do not need to know the state (or topology) of the system.
\item
\emph{Unilateral reinforcement:}
each user tends to increase his own utility;
put differently, the algorithm is aligned with each user's individual objective.
\item
\emph{Flexibility:}
the users' learning algorithm can be deployed in both static and ergodic (fast-fading) channel environments.
\end{enumerate}
As such, even though the static and ergodic channel regimes are fundamentally different, the network's users do not have to switch their update structure in order to converge to equilibrium (in the static or fast-fading regime, respectively).

Finally, our analysis is supplemented in Section \ref{sec:numerics} by extensive numerical simulations where we illustrate the throughput and power gains of the proposed approach under realistic conditions.

\section{System Model}
\label{sec:model}

Consider a set $\play = \{1,\dotsc,K\}$ of (unlicensed) \acfp{SU} that seek to connect to a common receiver over a set $\subs = \{1,\dots,S\}$ of non-interfering subcarriers (typically in the frequency domain if an \ac{OFDM} scheme is employed).
Focusing on the uplink case, the aggregate received signal $y_{\sub}$ over the $\sub$-th subcarrier will then be:
\begin{equation}
\label{eq:signal}
y_{\sub}
	= \insum_{k\in\play} h_{k\sub} x_{k\sub} + z_{\sub},
\end{equation}
where
\begin{enumerate}
\item
$x_{k\sub}\in\C$ denotes the transmitted signal of user $k\in\play$ over the $\sub$-th subcarrier.
\item
$h_{k\sub}\in\C$ is the corresponding transfer coefficient.
\item
$z_{\sub}\in\C$ denotes the aggregate interference-plus-noise received from all sources not in $\play$ (including the aggregate \ac{PU} transmission on subcarrier $\sub$ plus ambient and other peripheral interference effects);
throughout this paper (and by performing a suitable change of basis if necessary), we will model $z_{\sub}$ as a Gaussian variable $z_{\sub} \sim \mathcal{CN}(0,\noisevar_{\sub})$ for some positive $\sigma_{\sub}>0$.
\end{enumerate}
In this context, the average transmit power of user $k$ on subcarrier $\sub$ will be
\begin{equation}
p_{k\sub}
	= \ex\big[\vert x_{k\sub}\vert^{2}\big],
\end{equation}
where the expectation is taken over the (Gaussian) codebook of user $k$;
furthermore, each user's \emph{total} transmit power $p_{k} = \ex[\bx_{k}^{\dag} \bx_{k} ] = \sum_{\sub} p_{k\sub}$ will have to satisfy the power constraint
\begin{equation}
\label{eq:constraint-p}
p_{k}
	= \insum_{\sub\in\subs} p_{k\sub}
	\leq P_{k},
\end{equation}
where $P_{k}>0$ denotes the maximum transmit power of user $k\in\play$.
In this way, the set of admissible power allocation vectors for user $k$ is the $S$-dimensional polytope
\begin{equation}
\txs
\strat_{k}
	= \left\{\bp_{k}\in\R^{\subs}: p_{k\sub}\geq0 \text{ and } \insum_{\sub\in\subs} p_{k\sub} \leq P_{k} \right\},
\end{equation}
and the system's \emph{state space} (i.e., the space of all admissible power allocation profiles $\bp = (\bp_{1},\dotsc,\bp_{K})$) will be the product $\strat = \prod_{k}\strat_{k}$.

In this \acf{MC} framework, each user's achievable transmission rate depends on his individual \ac{SINR}
\begin{equation}
\label{eq:sinr}
\sinr_{k\sub}(\bp)
	= \frac{g_{k\sub} p_{k\sub}}{\noisevar_{\sub} + \sum_{\ell\neq k} g_{\ell\sub} p_{\ell\sub}},
\end{equation}
where $g_{k\sub} = \vert h_{k\sub} \vert^{2}$ denotes the channel gain coefficient for user $k$ over the $\sub$-th subcarrier.
Thus, in the \ac{SUD} regime (where interference by all other users is treated as additive noise), the maximum information transmission rate for user $k$ (achievable with random Gaussian codes) will be:
\begin{equation}
\label{eq:rate}
\rate_{k}(\bp)
	= \insum_{\sub\in\subs} \log\big( 1 + \sinr_{k\sub}(\bp) \big)
	= \insum_{\sub\in\subs} \left[
	\log\left( \noisevar_{\sub} + w_{\sub}(\bp) \right)
	- \log\left( \noisevar_{\sub} + \insum_{\ell\neq k} g_{\ell\sub} p_{\ell\sub}\right)
	\right]
\end{equation}
where
\begin{equation}
\label{eq:MUI}
w_{\sub}(\bp)
	= \insum_{k} g_{k\sub} p_{k\sub},
	\quad
	\sub=1,\dotsc,S,
\end{equation}
denotes the \emph{aggregate} \ac{SU} interference level per subcarrier (for convenience we will also write $\bw = (w_{1},\dotsc,w_{S})$ for the \acp{SU}' aggregate interference profile over all subcarriers $\sub\in\subs$).

In the absence of other considerations, the unilateral objective of each \ac{SU} would be the maximization of his individual transmission rate $\rate_{k}(\bp)$ subject to the total power constraint \eqref{eq:constraint-p}.
In our \ac{CR} setting however, the network operator needs to ensure that the system's \acp{PU} meet the \ac{QoS} guarantees that they have already paid for \textendash\ typically in the form of minimum rate requirements or maximum interference tolerance per subcarrier.
Thus, to achieve this, we will consider a general spectrum access pricing scheme whereby \acp{SU} are charged according to the individual and aggregate interference that they cause to the network's \acp{PU}.

Formally, this can be captured by the general cost model:
\begin{equation}
\label{eq:cost}
\cost_{k}(\bp)
	= \price_{0}(\bw(\bp))
	+ \price_{k}(\bp_{k}),
\end{equation}
where:
\begin{enumerate}
\item
$\price_{0}\from\R_{+}^{\subs}\to\R_{+}$ is a \emph{flat spectrum access price} that is calculated in terms of the aggregate \ac{SU} interference level $w_{\sub}$ per subcarrier $\sub\in\subs$.
\item
$\price_{k}\from\strat_{k} \to\R_{+}$ is a \emph{user-specific price} which is charged to user $k\in\play$ based on his \emph{individual} radiated power profile $\bp_{k}\in\strat_{k}$.
\end{enumerate}
In tune with standard economic considerations on diminishing returns \cite{MWG95}, the only assumptions that we will make for the price functions $\price_{0}$ and $\price_{k}$ are that:
\begin{enumerate}
[({A}1)\quad]
\item
Every price function $\price$ is non-decreasing in each of its arguments.
\item
Every price function $\price$ is Lipschitz continuous and convex.
\end{enumerate}
In particular, the convexity assumption (A2) acts as an interference control mechanism for the system:
by charging \acp{SU} higher spectrum access prices for the same increase in interference when the network operates in a high-interference state, \acp{SU} are implicitly encouraged to transmit at lower powers, thus creating less \acf{CCI} to the network's \acp{SU}.
In this way, the pricing scheme \eqref{eq:cost} is flexible enough to account for very diverse pricing paradigms:
if $\price_{0}\equiv0$, the network's \acp{SU} are charged on an equitable user-by-user basis, based only on the individual interference that each \emph{individual} user induces to the network's \acp{PU};%
\footnote{Likewise, $\price_{k}$ could also account for the actual cost incurred by the user to recharge the battery of his wireless device as in \cite{DMMP14}.}
otherwise, if $\price_{k}\equiv0$,
the pricing model \eqref{eq:cost} allows the network operator to reimburse infractions to the \acp{PU}' contractual \ac{QoS} guarantees by imposing an aggregate ``sanction'' to the network's \acp{SU} (who were responsible for causing the violation in the first place).

The specifics of the pricing functions $\price_{0}$ and $\price_{k}$ are negotiated between network users and operators based on their needs and means, so they can vary widely depending on the context \textendash\ see e.g. \cite{ABSA02,WCPW07,DMMP14,PSPF10,WJH14}.
For concreteness, we provide below some typical examples of pricing models which we explore further in Section \ref{sec:numerics}:%
\footnote{For simplicity, we focus on the flat-rate case;
the corresponding user-specific price functions $\price_{k}$ are defined similarly.}
\begin{enumerate}
[\itshape Model~1.\;]
\item
Let $\intf_{\sub}$ denote the \acp{PU}' interference tolerance on subcarrier $\sub$.
Then, in the spirit of \cite{PSPF10}, we define the \ac{LP} flat-rate model as:
\begin{equation}
\tag{LP}
\label{eq:LP}
\price_{0}^{\LP}(\bw)
	= \lambda_{0} \insum_{s\in\subs} w_{\sub}/\intf_{\sub},
\end{equation}
where the pricing parameter $\lambda_{0}$ represents the price paid by the network's \acp{SU} when saturating the \acp{PU}' interference tolerance.
In words, \acp{SU} are charged a flat-rate which is proportional to the degree of saturation of the \acp{PU}' interference tolerance level, so the model \eqref{eq:LP} treats the \acp{PU}' requirements as a soft constraint.

\item
With notation as above, the \ac{VP} flat-rate model is defined as:
\begin{equation}
\tag{VP}
\label{eq:VP}
\price_{0}^{\VP}(\bw)
	= \lambda_{0} \insum_{\sub\in\subs} \left[ w_{\sub}/\intf_{\sub} - 1 \right]_{+}
\end{equation}
where $\lambda_{0}>0$ is a sensitivity parameter and $[x]_{+} \equiv \max\{x,0\}$.
In this model, \acp{SU} are only charged when the \acp{PU}' interference tolerance is actually violated, and the steepness of the sanction is controlled by the pricing parameter $\lambda_{0}$;
as such, in the large $\lambda_{0}$ limit, \eqref{eq:VP} treats the \acp{PU}' requirements as a hard constraint with very sharp violation costs.
\end{enumerate}

In light of all this, the \emph{utility} of user $k$ is defined as:
\begin{equation}
\label{eq:pay}
\pay_{k}(\bp)
	= \rate_{k}(\bp) - \cost_{k}(\bp),
\end{equation}
i.e., $\pay_{k}(\bp)$ is simply the user's achieved transmission rate minus the cost reimbursed to the network operator in order to achieve it.
In turn, this leads to the \emph{cost-efficient throughput maximization game} $\game\equiv\game(\play,\strat,\pay)$, defined as follows:
\begin{enumerate}

\item
The game's \emph{players} are the system's \aclp{SU} $k\in\play=\{1,\dotsc,K\}$.

\item
The \emph{action set} of each player/user is the set of feasible power allocation profiles $\strat_{k} = \{\bp_{k}\in\R^{\subs}: p_{k\sub} \geq 0 \text{ and } \sum_{\sub\in\subs} p_{k\sub} \leq P_{k}\}$.

\item
Each player's \emph{utility function} $\pay_{k}\from\strat\equiv\prod_{k}\strat_{k}\to\R$ is given by \eqref{eq:pay}.
\end{enumerate}
In this context, we will say that a power allocation profile $\eq\in\strat$ is at \emph{\ac{NE}} when
\begin{equation}
\label{eq:Nash}
\tag{NE}
\pay_{k}(\eqvec_{k};\eqvec_{-k})
	\geq \pay_{k}(\bp_{k};\eqvec_{-k})
	\quad
	\text{for all $\bp_{k}\in\strat_{k}$ and for all $k\in\play$},
\end{equation}
i.e., when
\emph{each user's chosen power profile $\eq_{k}\in\strat_{k}$ is individually cost-efficient given the power profile of his opponents}
(so no user has a unilateral incentive to deviate).
Accordingly, our goal in the rest of the paper will be to characterize the Nash equilibria of $\game$ and to provide distributed optimization methods allowing selfish (and myopic) \acp{SU} to converge to equilibrium in the absence of centralized medium access control mechanisms.

\PM{Feel free to include a more elaborate discussion of Nash equilibria here if you deem it relevant.}


\section{Equilibrium Analysis, Learning and Convergence}
\label{sec:learning}

In this section, we focus on the characterization of the \acp{NE} of the cost-efficient rate maximization game $\game$ and on how players can attain such a state by means of a simple, adaptive learning process.

\subsection{Equilibrium structure and characterization}
\label{sec:characterization}

A key property of the rate maximization game $\game$ is that it admits a \emph{potential function} \cite{MS96}:

\begin{proposition}
\label{prop:potential}
Let $\bw$ be the aggregate \ac{SU} interference level defined as in \eqref{eq:MUI}.
Then, the function
\begin{equation}
\label{eq:potential}
\pot(\bp)
	= \insum_{\sub} \log\left(\noisevar_{\sub} + w_{\sub}\right)
	- \price_{0}(\bw)
	- \insum_{k} \price_{k}(\bp_{k})
\end{equation}
is an exact potential for the cost-efficient rate maximization game $\game$;
specifically:
\begin{equation}
\label{eq:potential-def}
\pay_{k}(\bp_{k};\bp_{-k}) - \pay_{k}(\bp_{k}';\bp_{-k})
	= \pot(\bp_{k};\bp_{-k}) - \pot(\bp_{k}';\bp_{-k})
\end{equation}
for all $\bp_{k},\bp_{k}'\in\strat_{k}$ and for all $\bp_{-k}\in\strat_{-k}\equiv\prod_{\ell\neq k} \strat_{k}$. 
\end{proposition}

\begin{IEEEproof}
By inspection.
\end{IEEEproof}

Since the price functions $\price_{0}$ and $\price_{k}$ are convex, the potential function $\pot$ is itself concave (though not necessarily strictly so; see below).
By Proposition \ref{prop:potential}, it then follows that maximizers of $\pot$ are \acp{NE} of $\game$ (so the Nash set of $\game$ is nonempty);
furthermore, with $\pot$ concave in $\bp$ and $\pay_{k}$ concave in $\bp_{k}$, every \acp{NE} of $\game$ is also a maximizer of $\pot$.
In this way, finding the equilibria of $\game$ boils down to the nonlinear optimization problem:
\begin{equation}
\label{eq:potential-max}
\begin{aligned}
\text{maximize}
	&\quad
	\pot(\bp),
	\\
\text{subject to}
	&\quad
	p_{k\sub}\geq0,
	\;\;
	\insum_{s} p_{k\sub} \leq P_{k}.
\end{aligned}
\end{equation}
Thanks to this formulation, we obtain the following equilibrium uniqueness result for $\game$:

\begin{theorem}
\label{thm:equilibrium}
Assume that:
\begin{enumerate}
[\upshape (C1)]
\item
Each user-specific price function $\price_{k}$ is strictly increasing in each of its arguments.
\\
or:
\item
The flat spectrum access price function $\price_{0}$ is either gentle enough or steep enough :
$0\leq \frac{\pd\price_{0}}{\pd w_{\sub}} < \left( \noisevar_{\sub} + \sum_{k} g_{k\sub} P_{k} \right)^{-1}$
or
$\frac{\pd\price_{0}}{\pd w_{\sub}} > 1/\noisevar_{\sub}$
for all $w_{\sub}$ and for all subcarriers $\sub\in\subs$.
\end{enumerate}
Then, the cost-efficient throughput maximization game $\game$ admits a unique Nash equilibrium for almost all realizations of the channel gain coefficients $g_{k\sub}$.
More generally, even if both \textup{(C1)} and \textup{(C2)} fail to hold, the set of Nash equilibria of $\game$ is a convex polytope of dimension at most $S(K-1)$.
\end{theorem}

\begin{IEEEproof}
See Appendix \ref{app:proofs-equilibrium}.
\end{IEEEproof}

\begin{remark}
The ``almost all'' part of the statement of Theorem \ref{thm:equilibrium} should be interpreted with respect to Lebesgue measure \textendash\ i.e., uniqueness holds except for a set of price functions and channel gain coefficients of Lebesgue measure zero.
In particular, if channel gains are drawn at the outset of the game following some fixed, continuous probability distribution (e.g., induced by the \acp{SU}' spatial distribution), then this means that $\game$ admits a unique equilibrium with probability $1$.
\end{remark}

\subsection{Exponential learning and convergence to equilibrium}
\label{sec:algorithm}

The equilibrium characterization of Theorem \ref{thm:equilibrium} is crucial from the standpoint of \ac{DSM} because it guarantees a very robust solution set (a convex polytope);
in fact, as we just saw, the game's equilibrium set is a singleton under fairly mild conditions for the users' price functions (e.g., that the user-specific price functions $\price_{k}$ be strictly increasing).
Regardless, given that it is far from clear how the system's users can compute the solution of the problem \eqref{eq:Nash}, our goal in this section will be to provide a distributed learning mechanism that can be employed by the system's users in order to reach a Nash equilibrium.

Our proposed algorithm will rely on the users' \emph{marginal utilities}:
\begin{equation}
\label{eq:marginal}
\payvec_{k}(\bp)
	= \nabla_{k}\pay_{k}(\bp)
\end{equation}
where $\nabla_{k}$ denotes differentiation with respect to the power profile $\bp_{k}$ of user $k$.
In particular, writing $\payvec_{k} = (\payv_{k,1},\dotsc,\payv_{k,S})$, some easy algebra yields the component-wise expression
\begin{equation}
\label{eq:marginal-coords}
\payv_{k\sub}(\bp)
	= \frac{\pd\pay_{k\sub}}{\pd p_{k\sub}}
	= g_{k\sub} \left(\frac{1}{\noisevar_{\sub} + w_{\sub}} - \frac{\pd\price_{0}}{\pd w_{\sub}} \right)
	- \frac{\pd\price_{k}}{\pd p_{k\sub}},
\end{equation}
which shows that $\payv_{k\sub}(\bp)$ can be calculated by each individual user knowing only their \ac{SINR} per subcarrier (which is measured locally) and the functional form of the price functions $\price_{0}$ and $\price_{k}$ (which are agreed upon by the network's \acp{SU} and the \ac{PU} and are thus also known locally).
Indeed, Eq.~\eqref{eq:sinr} shows that the aggregate interference level on subcarrier $\sub$ can be calculated by user $k$ as:
\begin{equation}
w_{\sub}(\bp)
	= \insum_{k} g_{k\sub} p_{k\sub}
	= g_{k\sub}p_{k\sub} + \insum_{\ell \neq k} g_{\ell\sub} p_{\ell\sub}
	= g_{k\sub} p_{k\sub} + \frac{g_{k\sub} p_{k\sub}}{\sinr_{k\sub}(\bp)}
	= g_{k\sub} p_{k\sub} \frac{1 + \sinr_{k\sub}(\bp)}{\sinr_{k\sub}(\bp)},
\end{equation}
i.e., requiring only local \ac{SINR} measurements and the knowledge of the user's channel (which can in turn be obtained through the exchange of pilot signals).
As a result, the marginal utility vectors $\payv_{k}$ can be calculated in a completely distributed fashion with locally available information.

By definition, the users' marginal utility vectors define the direction of unilaterally steepest utility ascent, i.e., the best direction that a user could follow in order to increase his utility.
As such, a natural learning process would be for each user to track this steepest ascent direction with the hopes of converging to a \acl{NE};
however, given the problem's power and positivity constraints, this method may quickly lead to inadmissible power profiles that do not lie in $\strat$ \textendash\ in which case convergence is also out of the question.

To account for these constraints, we will employ an interior point method which increases power on subcarriers that seem to be performing well, without ever shutting off a particular channel completely.
Formally, consider the \emph{exponential regularization map} $\gibbsvec\from\R^{\subs}\to\R_{+}^{\subs}$ given by
\begin{equation}
\label{eq:Gibbs}
\gibbsvec(\bv)
	= \frac{1}{1 + \insum_{\sub} \exp(\payv_{\sub})}
	\left( \exp(\payv_{1}),\dotsc,\exp(\payv_{S}) \right).
\end{equation}
This map has the property that it assigns positive weight (power) to all subcarriers and exponentially more weight to the subcarriers $\sub\in\subs$ with the highest marginal utilities $\payv_{\sub}$.
Furthermore, if all marginal utilities are relatively low (indicating high transmission costs), all assigned weights will also be low in order to decrease the user's cost.
With this in mind, our proposed exponential learning algorithm for cost-efficient rate maximization is as follows:

\begin{algorithm}[H]
{\sf
\small
\vspace{3pt}
Parameter:
step size $\step_{n}$.
\\[1pt]
Initialize:
$n \leftarrow 0$;
scores $\score_{k\sub} \leftarrow 0$ for all $k\in\play$, $\sub\in\subs$.
\\[1pt]
\Repeat
	{$n \leftarrow n+1$;
	\\[1pt]
	\ForEach{user $k \in \play$}
		{%
		\ForEach{subcarrier $\sub\in\subs$}
		{%
			set transmit power $\dis p_{k\sub} \leftarrow P_{k} \frac{\exp(\score_{k\sub})}{1 + \sum_{r} \exp(\score_{kr})}$;
			\\[1pt]
			measure $\sinr_{k\sub}$;
			\\[1pt]
			update marginal utilities:
			$\dis \payv_{k\sub}\leftarrow \frac{1}{p_{k\sub}} \frac{\sinr_{k\sub}}{1 + \sinr_{k\sub}} - \frac{\pd\cost_{k}}{\pd p_{k\sub}}$;
			\\[1pt]
			update scores:
			$\dis \score_{k\sub} \leftarrow \score_{k\sub} + \step_{n} \payv_{k\sub}$;
		} 
	} 
	\textbf{until} termination criterion is reached.
} 
} 
\caption{Exponential Learning for Cost-Efficient Rate Maximization}
\label{alg:XL}
\end{algorithm}

From an implementation point of view, Algorithm \ref{alg:XL} has the following desirable properties:
\begin{enumerate}[(P1)]
\item
It is \emph{distributed:}
users only need local or publicly available information in order to run it.
\item
It is \emph{stateless:}
users do not need to know the state of the system (e.g., its topology).
\item
It is \emph{reinforcing:}
users tend to allocate more power to cost-efficient subcarriers.
\end{enumerate}

We then obtain:

\begin{theorem}
\label{thm:conv}
Let $\step_{n}$ be a variable step-size sequence such that $\sum_{n} \step_{n} = \infty$ and $\sum_{j=1}^{n} \step_{j}^{2} \big/ \sum_{j=1}^{n} \step_{j} \to 0$.
Then, Algorithm \ref{alg:XL} converges to \acl{NE} in the cost-efficient rate maximization game $\game$.
\end{theorem}

\begin{IEEEproof}
See Appendix \ref{app:proofs-convergence}.
\end{IEEEproof}

\begin{remark*}
The condition $\sum_{j=1}^{n} \step_{j}^{2} \big/ \sum_{j=1}^{n} \step_{j} \to 0$ requires the use of a decreasing step-size $\step_{n}$ (which slows down the algorithm), but the rate of decay of $\step_{n}$ can be arbitrarily slow \textendash\ in stark contrast to the much more stringent requirement $\insum_{j}\step_{j}^{2}< \infty$ that is common in the theory of stochastic approximation \cite{Ben99}.
As such, Algorithm \ref{alg:XL} can be used with an effectively constant (very slowly varying) step-size, and still converge to equilibrium;
we explore this issue in detail in Section \ref{sec:numerics}.
\end{remark*}

\section{Fast-Fading and User Mobility}
\label{sec:fading}

Our analysis so far has focused on static channels, corresponding to wireless users with little or no mobility.
In this section, we investigate the case of mobile users where the channel gain coefficients evolve over time following a stationary ergodic process.

In this fast-fading regime, the users' achievable rate is given by the ergodic average \cite{GV97}:
\begin{equation}
\label{eq:rate-erg}
\bar\rate_{k}(\bp)
	= \ex_{g} \rate_{k}(\bp)
	= \insum_{\sub\in\subs} \ex_{g} \log\left(1 + \frac{g_{k\sub} p_{k\sub}}{\noisevar_{\sub} + \insum_{\ell\neq k} g_{\ell\sub} p_{\ell\sub}}\right),
\end{equation}
leading to the corresponding average utility functions:
\begin{equation}
\label{eq:pay-erg}
\bar\pay_{k}(\bp)
	= \bar\rate_{k}(\bp) - \ex_{g}[\cost_{k}(\bp)]
	= \bar\rate_{k}(\bp) - \ex_{g} \big[ \price_{0}(\bw) + \price_{k}(\bp_{k}) \big],
\end{equation}
where the expectation $\ex_{g}[\cdot]$ is taken with respect to the law of the channel gain coefficients $g_{k\sub} = \abs{h_{k\sub}}^{2}$ (recall here that the aggregate \ac{MUI} per subcarrier $w_{\sub} = \insum_{k\in\play} g_{k\sub} p_{k\sub}$ depends itself on the realization of the channels).
We thus obtain the following game-theoretic formulation of cost-efficient throughput maximization in the presence of fast fading:
\begin{equation}
\label{eq:game-erg}
\begin{aligned}
\text{maximize}
	&\quad
	\bar\pay_{k}(\bp_{k};\bp_{-k})
	\quad
	\text{(unilaterally for all $k\in\play$)},
	\\
\text{subject to}
	&\quad
	\bp_{k}\in\strat_{k}.
\end{aligned}
\end{equation}

As in the static regime, we then obtain the following characterization of Nash equilibria:

\begin{proposition}
\label{prop:potential-erg}
With notation as above, let
\begin{equation}
\label{eq:potential-erg}
\bar\pot(\bp)
	= \insum_{\sub} \ex_{g} \log\left(\noisevar_{\sub} + w_{\sub}\right)
	- \ex_{g} \left[ \price_{0}(\bw) + \insum_{k} \price_{k}(\bp_{k}) \right].
\end{equation}
Then, $\bar\pot(\bp)$ is an exact potential for the ergodic rate maximization game $\bar\game \equiv \bar\game(\play,\strat,\bar\pay)$.
In particular, if the channels' law is atom-free (i.e., it is absolutely continuous with respect to Lebesgue measure), $\bar\pot$ is strictly concave and $\bar\game$ admits a unique Nash equilibrium.
\end{proposition}

\begin{IEEEproof}
See Appendix \ref{app:proofs-fading}.
\end{IEEEproof}

Proposition \ref{prop:potential-erg} shows that the inherent stochasticity in the users' channels actually helps in guaranteeing a very robust solution set for the cost-efficient throughput maximization problem \eqref{eq:game-erg} (see also \cite{MBML12} for a related result in the context of rate control).
On the other hand, the expectation over the users' channels is typically hard to carry out (especially beyond the Gaussian \acs{iid} regime), so it is not clear how to calculate the ergodic marginal utilities $\bar\payvec_{k}(\bp) = \nabla_{k} \bar\pay_{k}(\bp) = \ex_{g}[\payvec(\bp)]$.
Thus, instead of trying to reach a \acl{NE} by employing a variant of Alg.~\ref{alg:XL} run with the users' ergodic marginal utilities (whose calculation requires considerable computation capabilities and a good deal of knowledge on the channels' statistics), we will consider the same sequence of events as in the case of static channels:
\begin{enumerate}
\item
At every update period $n=1,2,\dotsc$, each user $k\in\play$ calculates his instantaneous marginal utility vector $\payvec_{k}(n)$ following \eqref{eq:marginal-coords}:
\begin{equation}
\label{eq:marginal-inst}
\hat\payv_{k\sub}(n)
	= \frac{1}{p_{k\sub}(n)} \frac{\sinr_{k\sub}(n)}{1 + \sinr_{k\sub}(n)}
	- \left.\frac{\pd\cost_{k}}{p_{k\sub}}\right\vert_{\bp(n)}
\end{equation}
\item
Users update their powers following the recursion step of Alg.~\ref{alg:XL}, and the process repeats.
\end{enumerate}

Remarkably, despite the inherent stochasticity, we have:
\begin{theorem}
\label{thm:conv-erg}
Assume that the variance of the users' channel gain coefficients is finite.
If Alg.~\ref{alg:XL} is run with step-sizes $\step_{n}$ such that $\sum_{n} \step_{n} = \infty$ and $\sum_{j=1}^{n} \step_{j}^{2} \big/ \sum_{j=1}^{n} \step_{j} \to 0$, then the users' power profiles converge to \acl{NE} in the cost-efficient ergodic rate maximization game $\bar\game$ \textup(a.s.\textup).
\end{theorem}

\begin{IEEEproof}
See Appendix \ref{app:proofs-fading}.
\end{IEEEproof}

\begin{remark}
Thanks to Theorem \ref{thm:conv-erg}, we see that Algorithm \ref{alg:XL} enjoys the additional property:
\begin{enumerate}[(P1)]
\setcounter{enumi}{3}
\item
\emph{Flexibility:}
users can apply the algorithm ``as-is'' in both static and fast-fading environments.
\end{enumerate}
\end{remark}

\PM{Please feel free to include any remarks you deem relevant.}

\section{Numerical Results}
\label{sec:numerics}
To evaluate the performance of the proposed cost-efficient power allocation framework for throughput maximization in \acl{CR} networks, we have performed extensive numerical simulations over a wide range of system parameters.
In what follows, we provide a selection of the most representative cases.
\begin{table*}[t!]
   \begin{minipage}{0.48\textwidth}
   \centering
    \caption{\label{table:sim_setup}Simulation Setting}
    \begin{tabular}{|c|c|}
	    \hline
	    \textbf{Parameter} & \textbf{Value} \\
	    \hline
			Carrier frequency& $f_c=2.4\,\textrm{GHz}$ \\
			\hline
			Channel bandwidth	& $B=10.93\,\textrm{KHz}$ \\
			\hline
			Noise spectral density & $\noisedev_{\sub} = -173\,\textrm{dBm/Hz}$ \\
			\hline
			Maximum transmitting power of \acp{SU}	&  $P_{k} = 21.03\,\textrm{dBm}$\\
			\hline
			Edge of the simulated square area	&  $L=200\,\textrm{m}$\\
			\hline
			Transmitting power of the \ac{PU}	&  $P^{\PU}=30\,\textrm{dBm}$\\
			\hline
			Distance of the \ac{PU} from the receiver	& $d=50\,\textrm{m}$\\
			\hline
    \end{tabular}
  \end{minipage}
  \centering
  \begin{minipage}{0.51\textwidth}
  \centering
    \caption{\label{table:QoS_setup}\ac{PU}'s Requirements}
    \begin{tabular}{|c|c|}
	    \hline
	    \textbf{Data Rate} & \textbf{$\intf$} \\
	    \hline
			$12.8\,\textrm{KHz}$& $-68.3\,\textrm{dBm}$ \\
			\hline
			$16\,\textrm{KHz}$	& $-70\,\textrm{dBm}$ \\
			\hline
			$32\,\textrm{KHz}$ & $-75.6\,\textrm{dBm}$ \\
			\hline
    \end{tabular}
  \end{minipage}
\end{table*}

Throughout this section, and unless explicitly mentioned otherwise, we consider a population of $K=10$ \acp{SU} uniformly distributed over a square area and $S=10$ non-interfering subcarriers with channel gain coefficients $g_{k\sub}$ drawn according to the path-loss model for Jakes fading proposed in \cite{CCGHHKMMRX07};
the other relevant simulation parameters are summarized in Table \ref{table:sim_setup}.
For simplicity, we also assume that $\noisedev_{\sub}$ and $P_{k}$ are equal for all $\sub\in\subs$ and all $k\in\play$;
finally, we will assume that \acp{PU} have the same interference tolerance level $\intf_{\sub}$ over all subcarriers $\sub\in\subs$.

To begin with, we evaluate the impact of interference pricing on the \acp{SU}' behavior by introducing the \emph{violation index}
\begin{equation}
\label{eq:violation}
\VI_{\sub} = w_{\sub}/\intf,
\end{equation}
i.e., the amount of interference generated by \ac{SU}s on the $\sub$-th subcarrier relative to the \acp{PU}' tolerance.
Obviously, $\VI_\sub \leq 1$ means that the system's \acf{IT} requirements are not violated, whereas $\VI_\sub > 1$ indicates a violation of the \acp{PU}' 
contractual \ac{QoS} guarantees that will have to be reimbursed by the network's \ac{SU}s.
Accordingly, in Fig. \ref{fig:violation}, we plot the system's average violation index $\VI=1/|\subs|\insum_{\sub\in\subs} \VI_{\sub}$ as a function of the pricing parameter $\lambda_{0}$ for different values of the maximum interference tolerance level $\intf$ under the flat-rate pricing scheme $\price_0(\bw)$.
As can be seen, if the \acp{PU}' maximum interference tolerance level is low (i.e., $\intf$ is small), \acp{SU} violate the resulting \acl{IT} constraint only if the value of the price parameter $\lambda_0$ is also low. 
Thus, the \acp{PU}' \ac{QoS} guarantees are violated only in the ``soft pricing'' regime where the pricing parameter $\lambda_{0}$ is not high enough to safeguard the \acp{PU}' low interference tolerance.
On the other hand, if the cost incurred due to violations is high enough, no violations are performed:
our simulations show that under both the \ac{LP} and \ac{VP} models, there exists a threshold value of the cost parameter $\lambda_{0}$ such that  the violation index at the game's \ac{NE} is always less than one,
i.e., the interference generated by \acp{SU} on each subcarrier is never higher than the \acp{PU}' \ac{IT} constraints.

That being said, increasing the flat-rate pricing parameter $\lambda_{0}$ can lead to significantly different \ac{SU} behavior with respect to the \acp{PU}' interference tolerance level.%
\footnote{Recall here that, under \ac{VP}, the system's \acp{SU} are not charged when their aggregate interference $w_{\sub}$ is lower than $\intf$, and are (steeply) fined otherwise;
by contrast, the \ac{LP} model charges users even when the system's \ac{IT} constraints are not violated.}
In fact, under the \ac{LP} pricing model, \ac{SU} interference disincentives can become excessive:
Fig. \ref{fig:violation} shows that transmission costs for high $\lambda_{0}$ are so high (even for low interference levels) that \acp{SU} prefer to shut down and stop transmitting altogether.
On the other hand, under the \ac{VP} model, $\lambda_{0}$ affects the outcome of the game \emph{only} if the \acp{PU}' maximum interference tolerance is low:
increasing $\lambda_{0}$ beyond a certain value does not lead \acp{SU} to shut down and does not impact their sum-rate at equilibrium, precisely because \acp{SU} are charged only if they cause excessive interference to the system's \aclp{PU}.

To illustrate the system's transient phase when users employ Algorithm \ref{alg:XL} to optimize their utility, Fig. \ref{fig:wall} shows the aggregate interference on a given subcarrier when the interference constraint is set to $\intf=-70\,\textrm{dBm}$ and users are charged based on the \ac{VP} flat-rate model.
We see there that the \ac{PU}'s interference constraint is violated only during the first few iterations of the learning process:
when the interference in a given subcarrier exceeds the \acp{PU}' tolerance, the \acp{SU} experience a sharp drop in their marginal utilities \eqref{eq:marginal-coords} because of the incurred cost $\price_{0}^{\VP}(\bw)$, so Algorithm \ref{alg:XL} prompts them to reduce their radiated power in the next iteration in order to avoid further violations.
In this way, \ac{SU} violations are quickly reduced and the users' learning process converges to a violation-free \acl{NE} of the cost-efficient throughput maximization game.

In Fig.~\ref{fig:sum_rateSU}, we evaluate the impact of pricing and power constraints on the system's performance at \acl{NE} for different pricing models.
Under the \ac{VP} model, the \acp{SU}' sum-rate at equilibrium is affected by the cost parameter $\lambda_0$ only when $\lambda_{0}$ is small:
the reason for this is that \acp{SU} do not violate the \acp{PU}' \acl{IT} constraints for high $\lambda_{0}$ (cf. Fig.~\ref{fig:violation}), so their transmit power and sum-rate at equilibrium remains (almost) constant for high $\lambda_{0}$.
On the other hand, as in the case of Fig.~\ref{fig:violation}, Fig. \ref{fig:sum_rateSU} shows that the \ac{LP} model (solid lines) is strongly affected by the pricing parameter $\lambda_{0}$, for all $\lambda_{0}$ values:
since increasing $\lambda_{0}$ in the \ac{LP} model increases transmission costs across the board, each \ac{SU} is pushed to reduce his individual transmit power in order to reduce the induced mutual interference in the network commensurately.
It is worth noting however that increasing transmission costs is \emph{not} always detrimental to \acp{SU} under the \ac{LP} model:
as shown in Fig.~\ref{fig:sum_rateSU}, there is a pricing parameter region where the overall interference on a given channel decreases when $\lambda_{0}$ is increased, thus enabling users to achieve higher data rates (due to the decreased interference on the channel). 
Nonetheless, in the presence of much higher transmission costs, the radiated power of \acp{SU} is too low to carry any significant amount of information, thus leading to a decrease in achievable throughput.

We also show the impact of different system configurations on the achievable \ac{SU} performance by plotting the users' average sum-rate at equilibrium for different values of the system's \emph{congestion index}, i.e., the ratio $K/S$ between the number of \acp{SU} accessing the system and the number of available subcarriers.
As expected, networks with low congestion (i.e., $K/S={0.5,1}$) exhibit better performance than highly congested networks (i.e., $K/S=1.5$):
when there is a higher number of \acp{SU} trying to access the network, the mutual interference also increases, thus causing considerable losses in throughput and leading \acp{SU} to shut down instead of incurring high transmission costs for moderate-to-low gains in throughput.

\begin{figure*}[t!]
   \begin{minipage}{0.49\textwidth}
    \includegraphics[width=\textwidth]{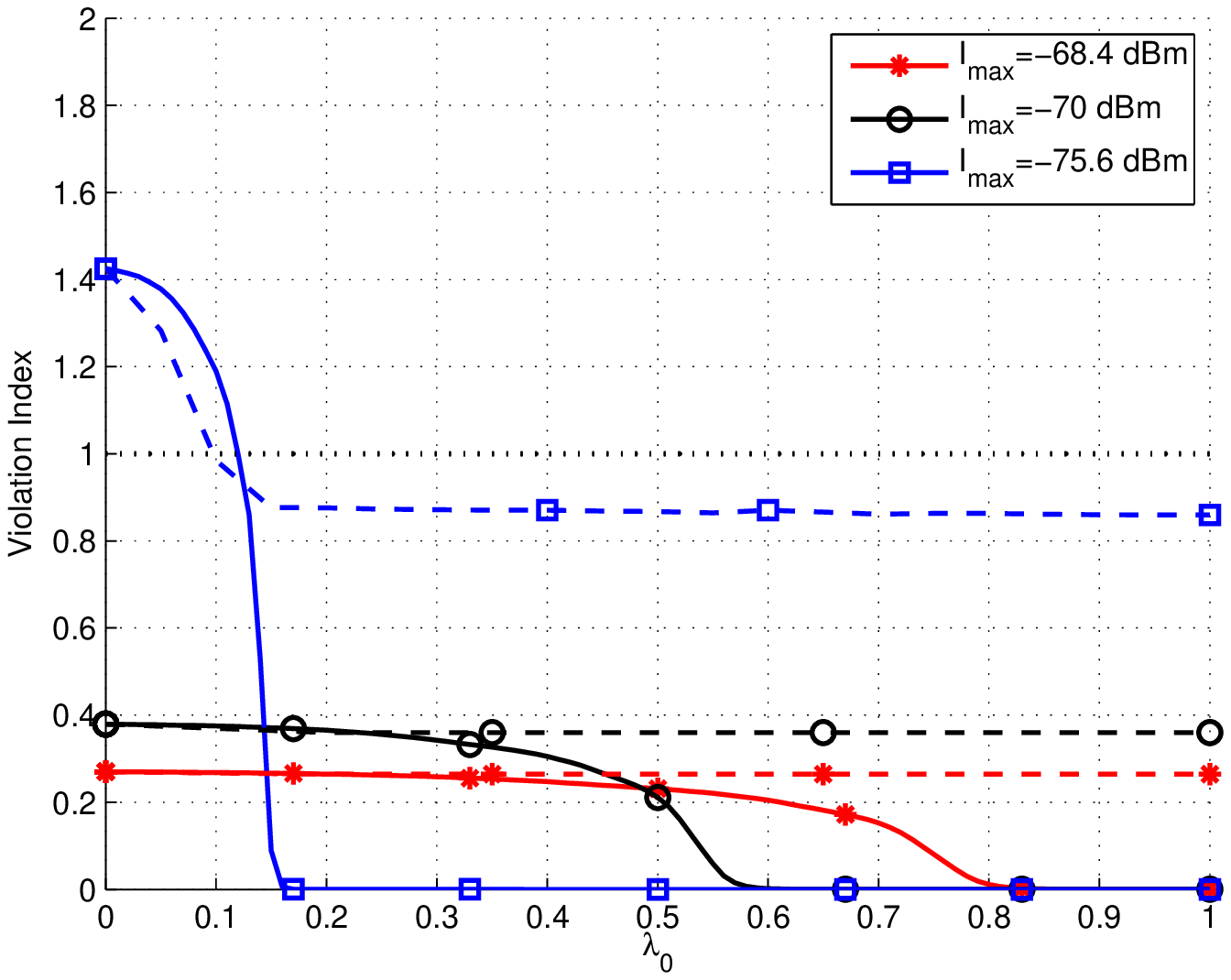}
    \vspace{-1cm}
    \caption{\label{fig:violation} Violation index as a function of $\lambda_0$ for different values of the maximum \acl{IT} level $\intf$ under the flat-rate pricing schemes (\ac{LP}: solid lines; \ac{VP}: dashed lines).}
  \end{minipage}
  \hspace{0.02\textwidth}
  \begin{minipage}{0.49\textwidth}
    \includegraphics[width=\textwidth]{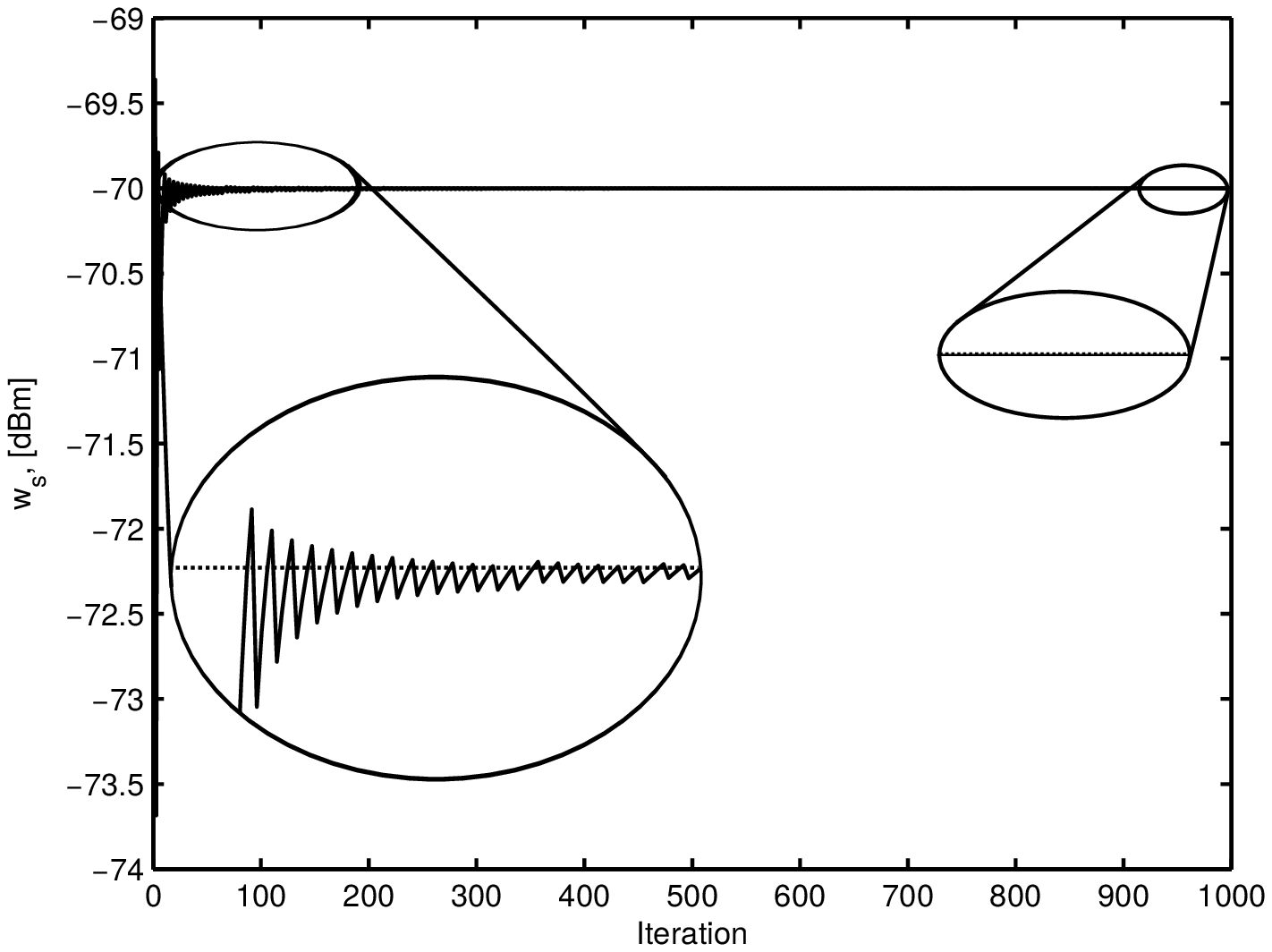}
    \vspace{-1cm}
    \caption{\label{fig:wall} Impact of the interference constraint on the evolution of the learning process under the \ac{LP} model, ($\intf_{\sub}=-70\,\textrm{dBm}$).}
  \end{minipage}
\end{figure*}

\begin{figure*}[t]
   \begin{minipage}{0.49\textwidth}
    \includegraphics[width=\textwidth]{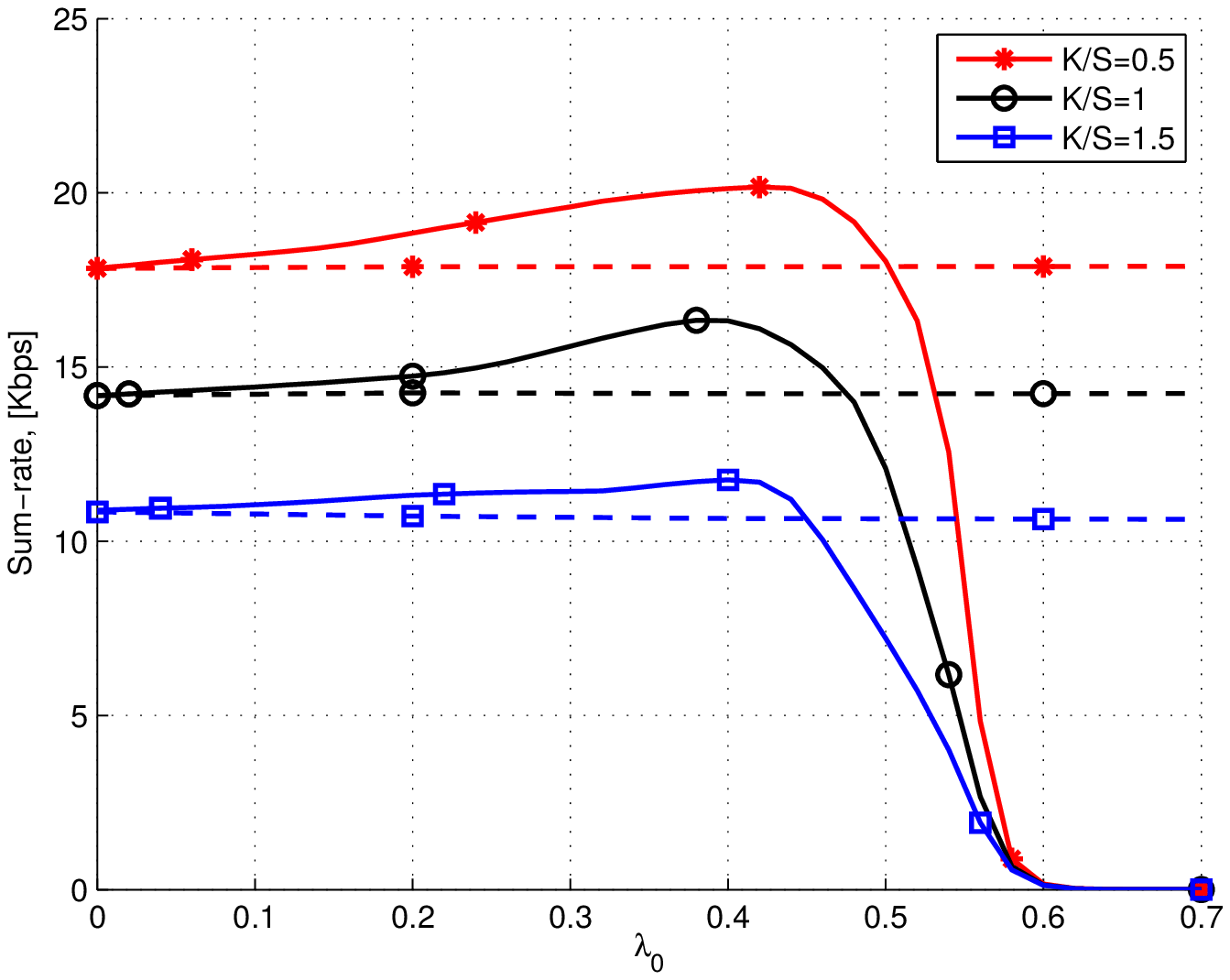}

    \caption{\label{fig:sum_rateSU} Average sum-rate as a function of different pricing models, system configurations and values of the pricing parameter $\lambda_0$ (\ac{LP}: solid lines; and \ac{VP}: dashed lines).}
  \end{minipage}
  \hspace{0.02\textwidth}
  \begin{minipage}{0.49\textwidth}
    \includegraphics[width=\textwidth]{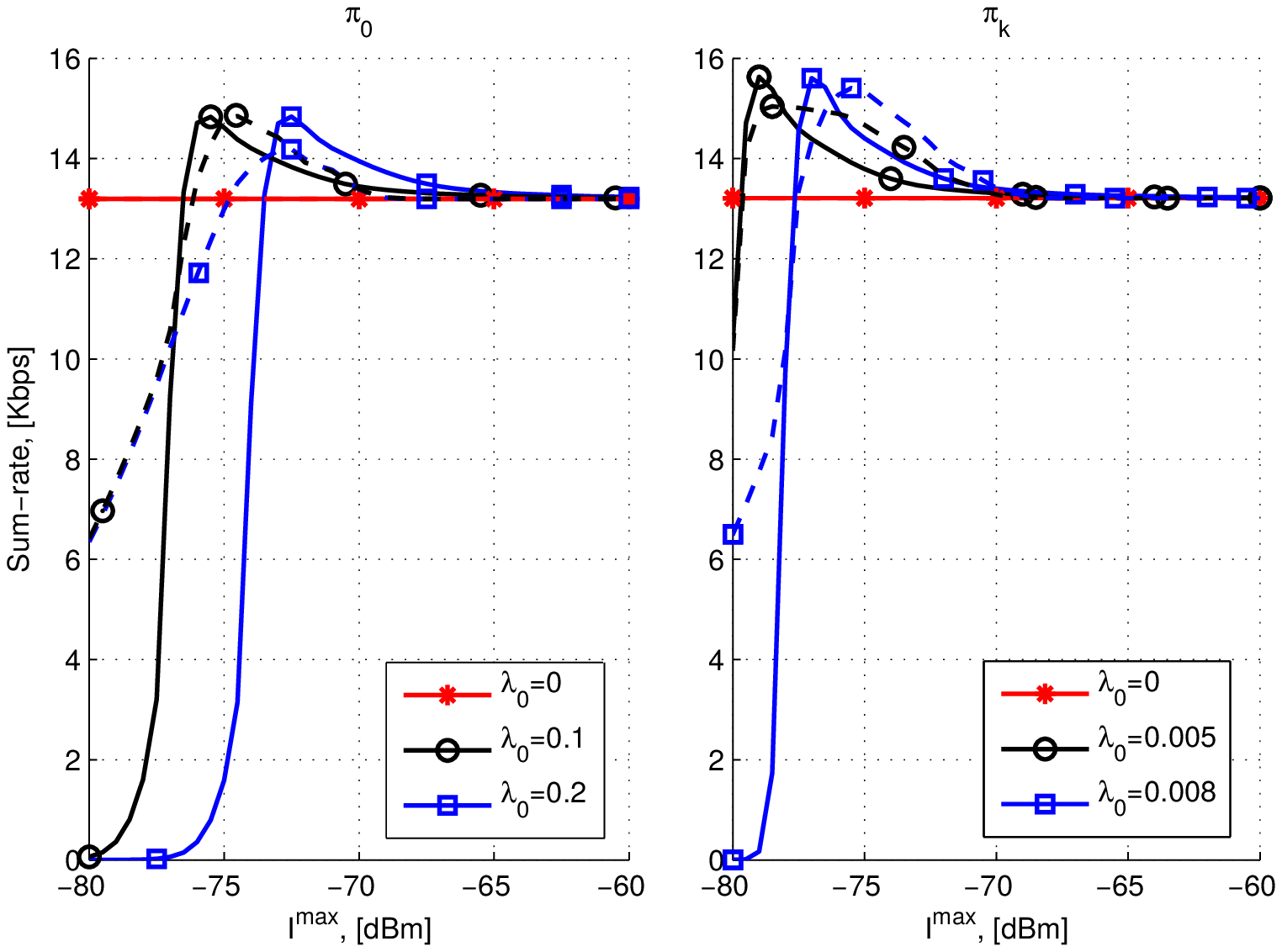}
    \vspace{-1cm}
    \caption{\label{fig:sum_rateINT} Average sum-rate as a function of the maximum interference $\intf$ at the \ac{PU} for different pricing schemes and values of the pricing parameters $\lambda_0$ under the \ac{LP} model  (\ac{LP}: solid lines; \ac{VP}: dashed lines).}
  \end{minipage}
\end{figure*}

In Fig.~\ref{fig:sum_rateINT} we illustrate how the \acp{SU}' sum-rate at equilibrium varies as a function of the \acp{PU}' interference tolerance $\intf$ for different pricing schemes (linear vs. violation pricing and flat-rate vs. per-user pricing).
Obviously, when \ac{SU} transmission comes at no cost (the $\lambda_{0}=0$ case), the value of $\intf$ does not impact the outcome of the game.
On the other hand, when $\lambda_{0} > 0$,
the \acp{SU}' average sum-rate increases as the \acp{PU}' interference tolerance increases up to a critical value $\intf_{c}$ where the \acp{SU}' sum-rate achieves its maximum value.
For any tolerance level $\intf > \intf_{c}$, the \acp{SU}' average sum-rate starts decreasing and eventually converges to a well-defined limit value as $\intf\to\infty$, corresponding to the case where the \ac{PU} is allowing free access to the leased part of the spectrum.
This occurrence is similar to what we have already discussed in Fig.~\ref{fig:sum_rateSU} and stems from the fact that low prices (small $\lambda_{0}$) and/or high interference tolerance (large $\intf$) do not provide a strong disincentive for \acp{SU} to reduce their power level;
as a result, the mutual interference across \acp{SU} also increases and leads to a decrease in the achievable performance of the secondary network.
\PM{Would it be possible to use the same scale for the y-axis in Fig.~\ref{fig:sum_rateINT}?
Also, I changed slightly the end of this paragraph.}
\SD{Done!}
Importantly, when $\intf$ is relatively low, the \ac{LP} and \ac{VP} models exhibit different behaviors, illustrated by the fact that the \acp{SU}' sum-rate at equilibrium differs.
By contrast, \eqref{eq:LP} and \eqref{eq:VP} both tend to zero as $\intf\to\infty$, so their behavior for very large $\intf$ is similar and the system converges to the same sum-rate value.

\PM{Added the following paragraph to try to explain the maximum that we see.}
The observed sum-rate maximum for intermediate values of $\intf$ can be explained as follows:
in the intolerant regime (small $\intf$), users hardly transmit at all because of the \acp{PU}' strict \ac{QoS} requirements;
on the other hand, in the ``open network'' regime (large $\intf$), each user selfishly transmits at maximum power in order to maximize his individual throughput (since there is no cost balancing factor), thus increasing interference and reducing the users' sum-rate (in a manner similar to the classical prisoner's dilemma).
As a result, the \acp{SU}' sum-rate is maximized for an intermediate value of $\intf$ where \acp{SU} have to control their power in order to avoid being charged for \ac{IT} violations:
in other words, a proper choice of $\intf$ (or, equivalently, $\lambda_{0}$) allows \acp{SU} to achieve a state which is both unilaterally stable and Pareto efficient (in the sense described above).



\SD{1) I changed the comments to Fig.~\ref{fig:sum_rateINT} and 2) I guess that such behavior derives first from the simulated environment which is really prohibitive for SUs 
(i.e., channel gains are really lossy, probably we would not see the same under a different simulation setup)
and from Fig.~\ref{fig:sum_rateSU} (even Fig. \ref{fig:rev_rate} has a similar behavior). In fact, the price that SUs pay for their interference
is $\propto \lambda_0 /\intf$, thus the price increases by either increasing $\lambda_0$, or decreasing $\intf$.
Basically, what you do is to take Fig.~\ref{fig:sum_rateSU} and reflect it in front of a mirror.
Here we have a similar (not the same, note that Fig.~\ref{fig:sum_rateSU} does not give us any info on the relationship between $\lambda_0$ and $\intf$ as we have a fixed $\intf=70\,\textrm{dBm}$, so we lack some information) 
behavior but inverted due to the relation between $\lambda_0$ and $\intf$.
Here it should hold that low maximum interference levels (or high values of $\lambda_0$) push SUs in turning off their transmitters (i.e., low rates).
By increasing (i.e., relaxing) $\intf$ (or decreasing the cost $\lambda_0$) SUs start transmitting some data as they experience reasonable costs. 
They reach a maximum sum-rate where the trade-off between interference cost, achieved sum-rate and mutual interference
is optimal, after that maximum point, SUs start increasing the radiated power which causes performance losses as the channel start being noisy, i.e., the mutual interference is high.
This is my personal interpretation of the phenomenon. However, we should discuss about it before the submission. Do you all have any other explanation?}
\PM{I took a crack at it and added a paragraph above your comment, let me know what you think.}
\SD{Definitely agree}


Finally, in Fig. \ref{fig:sum_rateINT} we also investigate the difference between flat-rate pricing ($\price_0$) and per-user pricing ($\price_k$) models. 
Both models exhibit similar properties, but for noticeably different values of $\lambda_{0}$:
specifically, to achieve the same sum-rate under per-user pricing, lower values of $\lambda_{0}$ should be considered, because users are much more sensitive to the value of $\lambda_{0}$ in the per-user paradigm.

%

In Fig. \ref{fig:rev_rate} we illustrate the transmission rate and revenue achieved by the \ac{PU} as a function of the pricing parameter $\lambda_0$ for different values of $\intf$ under the \ac{LP} and \ac{VP} schemes.
Specifically, the \ac{PU}'s sum-rate is calculated as
\begin{equation}
\rate_{\PU}(\bw)
	=\sum_{\sub\in\subs} \log\left( 1+\sinr^{\PU}_{\sub}(w_{\sub})  \right),
\end{equation}
where
$\sinr^{\PU}_{\sub}(w_{\sub})= g^{\PU} P^{\PU}/w_{\sub}$ is the \ac{PU}'s \ac{SINR} on the $\sub$-th subcarrier, and $g^{\PU}$ and $P^{\PU}$ denote the \ac{PU}'s channel gain and transmit power, respectively;
by the same token, the revenue of the \ac{PU} is simply $K\price_{0} + \insum_{k} \price_{k}$, i.e., the sum of the charges paid by the \acp{SU}.
For comparison purposes, we have fixed three different values of the parameter $\intf$ according to different \ac{PU} minimum data rate requirements (cf. Table \ref{table:QoS_setup}).

Importantly, as far as the \ac{LP} model is concerned, Fig. \ref{fig:rev_rate} shows that a high pricing parameter $\lambda_0$ brings \emph{no} revenue to the \ac{PU} because it acts as a severe transmission disincentive to the \acp{SU} (cf. Fig.~\ref{fig:sum_rateSU}, where we saw that \acp{SU} shut down beyond a certain threshold value $\lambda_{0}^{\ast}$).
Because of this behavior, there exists a critical value $\lambda_0^{c}$ for the pricing parameter that maximizes the \acp{PU}' revenue:
the calculation of this critical value lies beyond the scope of this paper, but it is evident that $\lambda_{0}^{c}$ increases when the maximum tolerable interference $\intf$ imposed by \acp{PU} also increases.
On the other hand, the \acp{PU}' revenue under the \ac{VP} model is almost always zero (or close to zero):
the reason for this is that the \ac{VP} model acts as a soft barrier (which hardens in the large $\lambda_{0}$ limit), so users tend to respect the \acp{PU}' requirements and thus incur no transmission-related penalties.
In other words, we see that if the \ac{PU}'s \ac{QoS} requirements are not too sharp, then the \ac{LP} model acts as a good source for revenue;
otherwise, if the \ac{PU}'s rate requirements are tight, the \ac{VP} model guarantees that \acp{SU} will respect them but does not generate any income.
Also, note that under both the \ac{LP} and \ac{VP} models, the rate of the \ac{PU} is always equal or higher than his minimum required data rate (dotted lines).
This is an important result that shows that pricing regulates the \acp{SU}' behavior indirectly (based on the \ac{PU}'s \ac{QoS} requirements and revenue targets), simply by fine-tuning the exact pricing model and its parameters (e.g., $\lambda_{0}$).

\begin{figure}[t]
\centering
\includegraphics[width=.45\columnwidth]{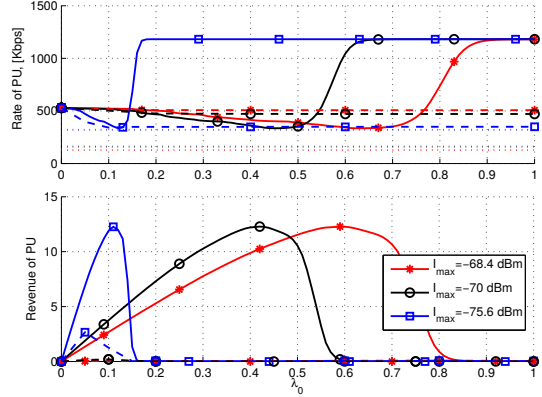}
\caption{Sum-rate and revenue of the \ac{PU} and total transmitting power of \ac{SU}s as a function of $\lambda_0$ for different values of the maximum \acl{IT} level $\intf$ under different pricing schemes (\ac{LP}: solid lines; \ac{VP}: dashed lines; Minimum data rate: dotted lines).}
\label{fig:rev_rate}
\end{figure}

Figs. \ref{fig:comparison1}\textendash\ref{fig:comparison3} compare the performance of the proposed power allocation scheme to the benchmark case of uniform power allocation \textendash\ i.e., when \acp{SU} transmit at full power and allocate their power uniformly over the available subcarriers, irrespective of the \ac{PU}'s requirements.
For some values of $\lambda_{0}$, the \acp{SU}' sum-rate under uniform power allocation is higher than the one achieved by the proposed approach, but this comes at the expense of violating the \ac{PU}'s minimum \ac{QoS} requirements (which constitutes a contractual breach from the \ac{PU}'s perspective);
on the contrary, our approach always respects the \ac{PU}'s contractual requirements (since the $\lambda_{0}$ pricing parameter is negotiated with the \ac{PU}), while guaranteeing high throughput to the \acp{SU}.
This is seen in Fig. \ref{fig:comparison2}:
the \ac{PU}'s throughput exceeds the throughput achieved when \acp{SU} employ a uniform power allocation policy, except when the \ac{PU} has no significant \ac{QoS} requirements ($\intf\to\infty$), in which case the \acp{SU} exploit all the available spectrum and the \ac{PU}'s rate is reduced.
Furthermore, in Fig. \ref{fig:comparison3} we illustrate the normalized revenue of the proposed approach w.r.t. the revenues generated by uniform power allocation policies.
\PM{Hmmm, I'm still not sure about the word ``efficiency'', I think it would be better to say ``normalized revenue'' or something along those lines.}
\SD{Fixed!}
Note that the income generated by the proposed approach is up to $3\times$ higher than the income generated by \acp{SU} that are not cost-/energy-aware and transmit naïvely at full power, using a uniform power allocation policy.%
\footnote{Recall here that the \ac{VP} model does not generate any revenue so, to reduce clutter, the corresponding curves are not shown.}
Thus, by fine-tuning his pricing scheme, the \ac{PU} not only achieves his \ac{QoS} requirements, but also increases his monetary revenue against cost-aware \acp{SU}.
%
\begin{figure*}[t]
\centerline{
\subfigure[]{
\includegraphics[width=.3\textwidth]{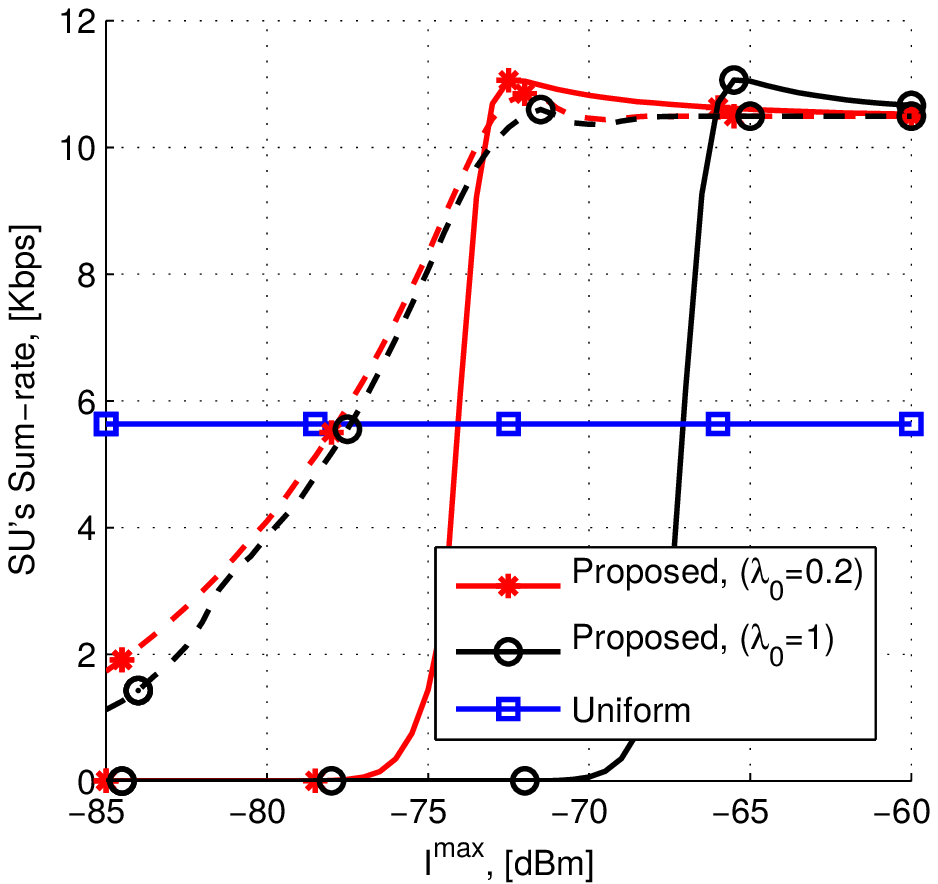}\label{fig:comparison1}
}
\hfil
\subfigure[]{
\includegraphics[width=.3\textwidth]{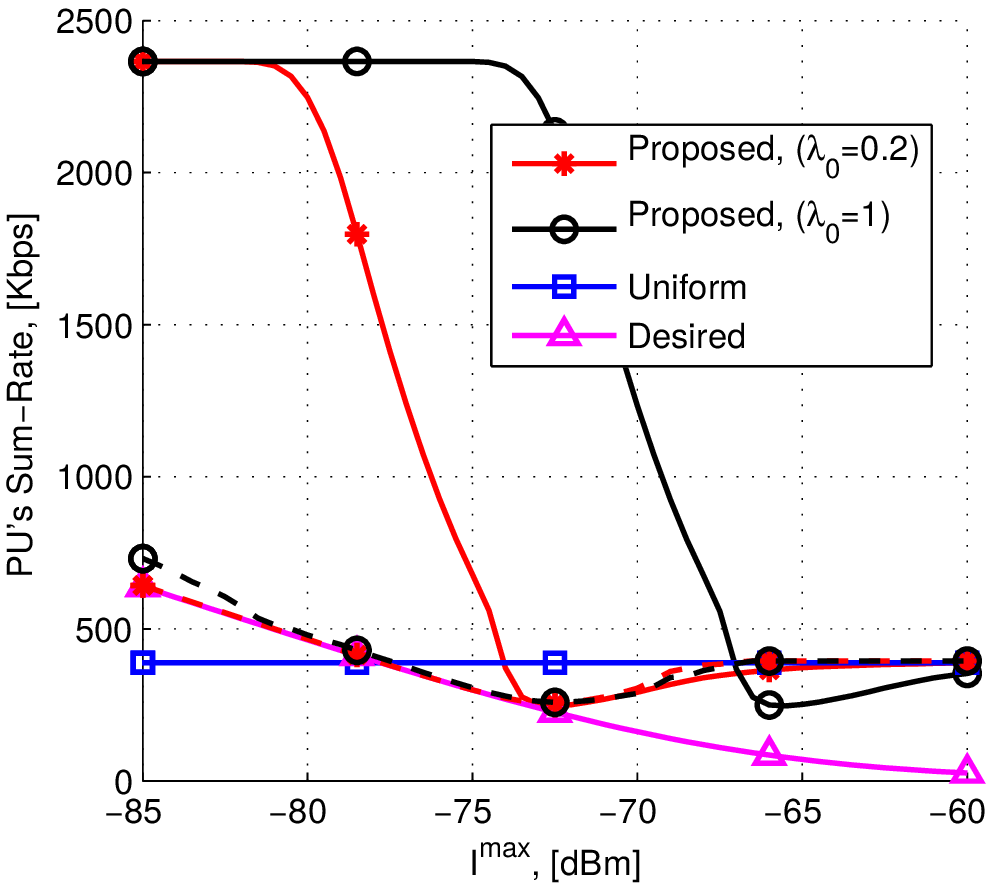}\label{fig:comparison2}
}
\hfil
\subfigure[]{
\includegraphics[width=.3\textwidth]{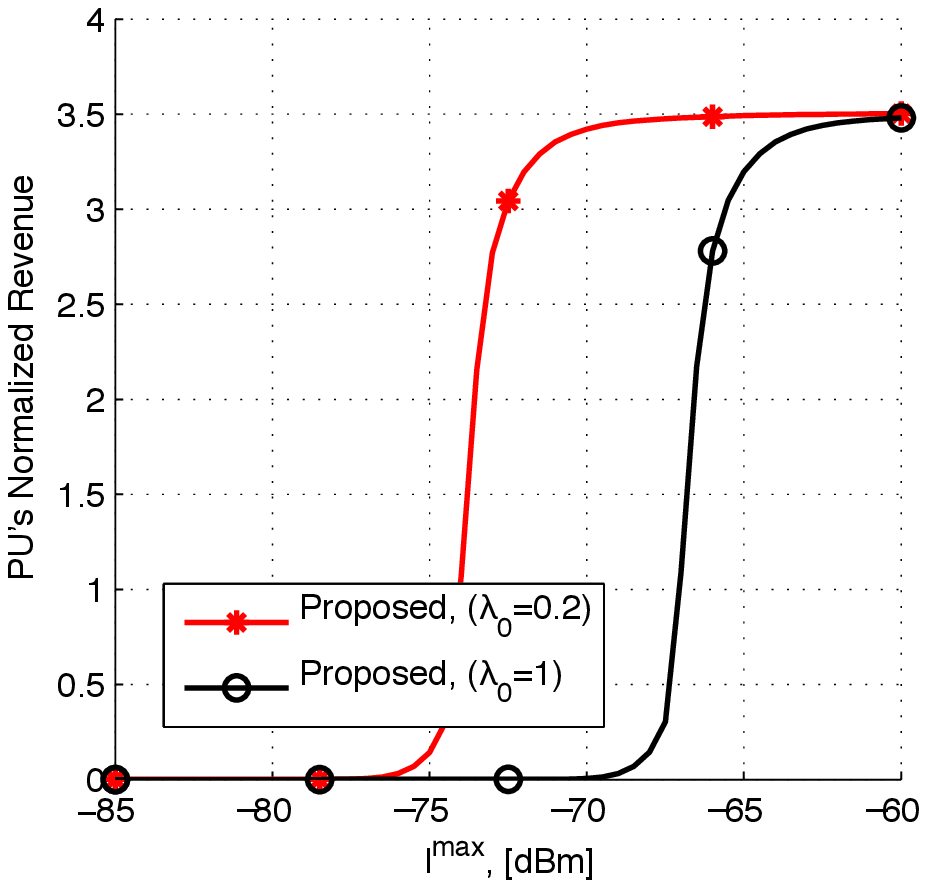}\label{fig:comparison3}
}}
\caption{Comparison between the proposed and uniform power allocation approaches: a) Sum-rate of \acp{SU}; b) Sum-rate of the \ac{PU}; c) Normalized revenue of the proposed approach w.r.t. the uniform power allocation policy (\ac{LP}: solid lines with star and circle markers; \ac{VP}: dashed lines with star and circle markers).}
\end{figure*}

In Figs.~\ref{fig:convergence} and \ref{fig:scalability}, we investigate the length of the system's off-equilibrium phase and the convergence rate of the proposed distributed learning scheme (Algorithm \ref{alg:XL}).
By Theorem \ref{thm:conv}, the iterations of Algorithm \ref{alg:XL} converge to \acl{NE} when using a step-size sequence $\step_{n}$ such that $\sum_{j=1}^{n}\step_{j}^{2} \big/ \sum_{j=1}^{n} \step_{j} \to 0$ as $n\to\infty$. 
As discussed in \cite{DMMP14}, a rapidly decreasing step-size sequence slows down the algorithm, so we examine here the usage of a fixed step size to accelerate convergence.
This choice makes the algorithm run faster;
on the other hand, a fixed step-size may lead to unwanted oscillations around the equilibrium point, thus interfering with the algorithm's end-state.
To account for this, we employ an adaptive \ac{STC} approach \cite{DM92}:
we start with a large, constant step-size which is then decreased as soon as oscillations are detected.%
\footnote{Note that such a step-size schedule still satisfies the summability postulates of Theorem \ref{thm:conv}.}
By means of this approach, Algorithm \ref{alg:XL} is very aggressive during the first non-oscillating iterations and it becomes more conservative (thus guaranteeing convergence) once oscillations are noticed.

To assess the method's efficiency, we plotted the system's \ac{EQL} defined as:
\begin{equation}
\eql(n)
	= \frac{\pot_{n}-\pot_{\min}}{\pot_{\max}-\pot_{\min}}
\end{equation}
\noindent
where $\pot_{n} \equiv \pot(\bp(n))$ is the potential \eqref{eq:potential} of the game at the $n$-th iteration of the algorithm, and $\pot_{\min}$ ($\pot_{\max}$) is the minimum (maximum) value of $\pot$;
obviously, an \ac{EQL} value of $1$ means that the system is at Nash equilibrium.
Accordingly, in Fig. \ref{fig:convergence}, we show the evolution of the \ac{EQL} and the system's sum-rate at each iteration for different step-size rules and interference pricing models.
As expected, a conservative step-size of the form $\step_{n} = n^{-\beta}$, $1/2<\beta<1$, leads to relatively slow convergence (of the order of several tens of iterations or worse).
On the other hand, the use of \ac{STC} and fixed-step methods greatly accelerates the users' learning rate:
after only a few \ac{STC} iterations the system's \ac{EQL} exceeds 90$\%$, and the algorithm's convergence is accelerated even further by increasing the constant step-size in the ``exploration'' phase of the \ac{STC} method.

\begin{figure*}[t!]
   \begin{minipage}{0.49\textwidth}
    \includegraphics[width=\textwidth]{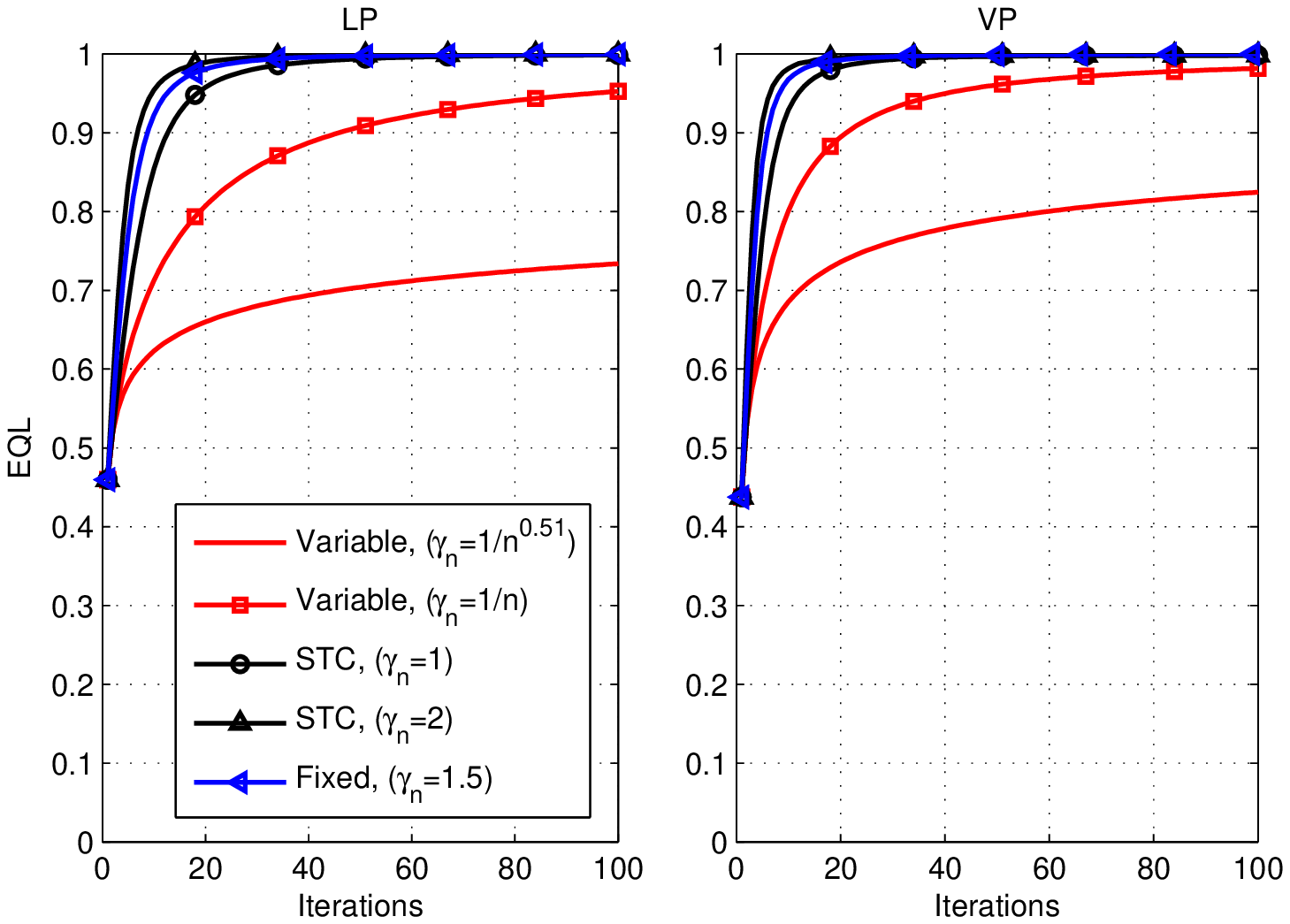}
    \vspace{-1cm}
    \caption{\label{fig:convergence}Equilibration level, $\eql(n)$, for different step-size rules under and flat-rate interference pricing models.}
  \end{minipage}
  \hspace{0.02\textwidth}
  \begin{minipage}{0.49\textwidth}
    \includegraphics[width=\textwidth]{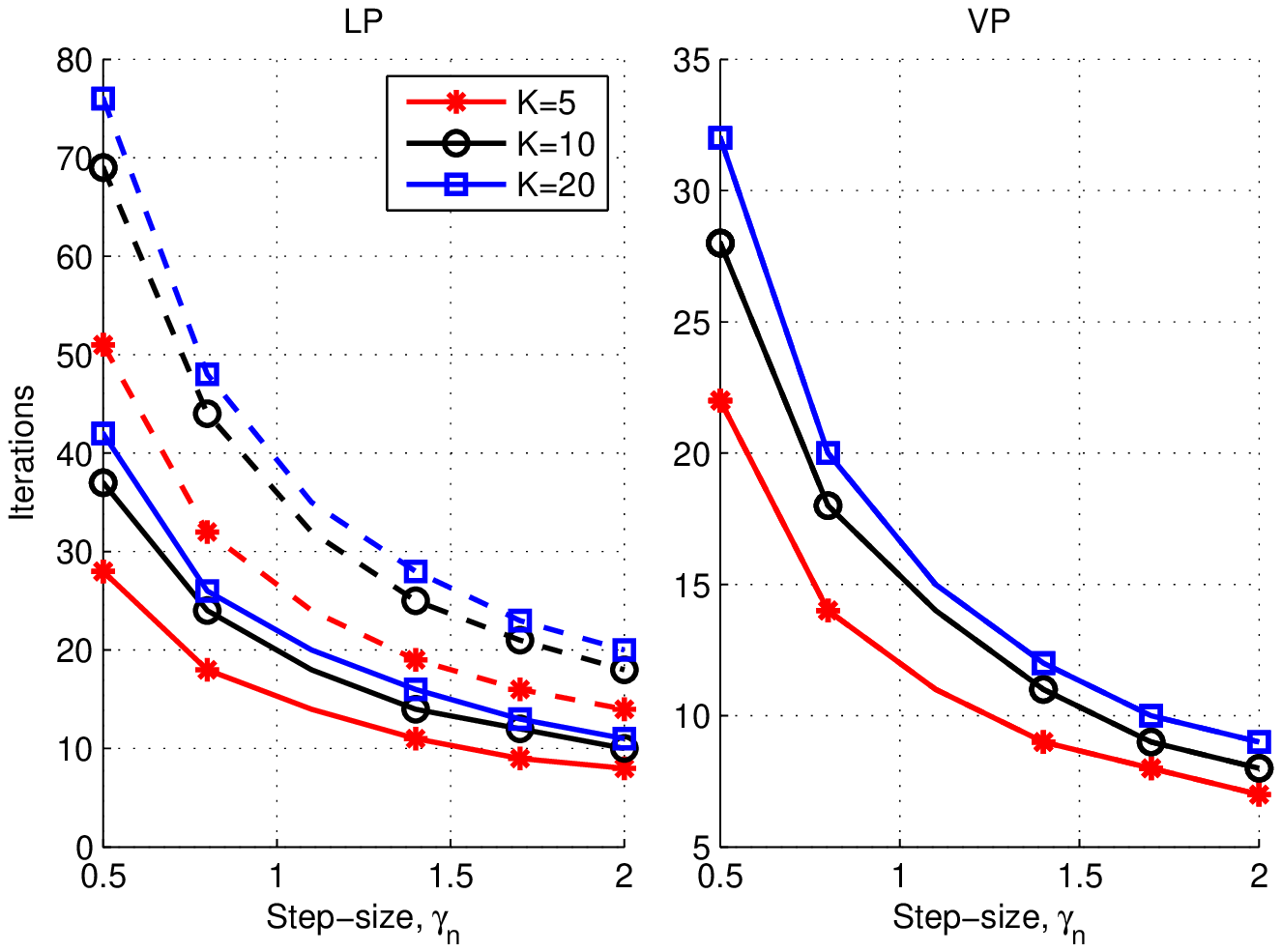}
    \vspace{-1cm}
    \caption{\label{fig:scalability}Scalability of the proposed learning scheme as a function of the step-size $\step$ for different values of the number $K$ of \ac{SU}s and pricing schemes ($\lambda_0=0.1$: solid lines; $\lambda_0=0.5$ dashed lines).}
  \end{minipage}
\end{figure*}

To investigate the scalability of the proposed learning scheme, we also examine the algorithm's convergence speed for different numbers of \acp{SU}.
In Fig. \ref{fig:scalability} we show the number of iterations needed to reach an \ac{EQL} of $95\%$:
importantly, by increasing the value of the algorithm's step-size, it is possible to reduce the system's transient phase to a few iterations, even for large numbers of users.
Moreover, we also note that the algorithm's convergence speed in the \ac{LP} model depends on the pricing parameter $\lambda_0$ (it decreases with $\lambda_{0}$), whereas this is no longer the case under the \ac{VP} model.
The reason for this is again that the \ac{VP} model acts as a ``barrier'' which is only activated when the \acp{PU}' interference tolerance is violated.

Finally, to investigate the impact of mobility and channel fading on the users' learning process, we consider a system with three \acp{SU} ($K=3$) and three \ac{iid} Gaussian fast-fading orthogonal subcarriers ($S=3$).
In Fig. \ref{fig:ergodic}, we plot the system's \ac{EQL} with respect to the ergodic potential \eqref{eq:potential-erg} under the \ac{LP} model as a function of different price settings and step-size rules.
Remarkably, even in this stochastic setting, Algorithm \ref{alg:XL} still converges to the game's \ac{NE} in a few iterations and, as before, the algorithm's convergence rate is improved by choosing more aggressive step-size sequences.

\begin{figure}[t]
\centering
\includegraphics[width=.45\columnwidth]{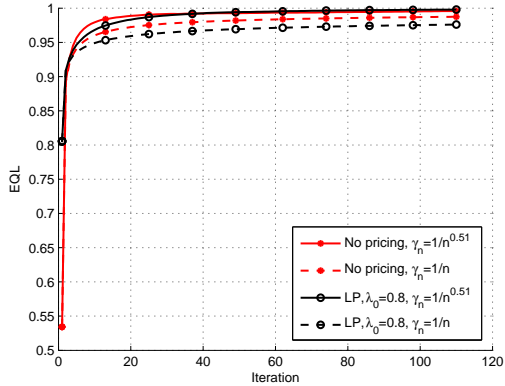}
\caption{Equilibration level (\ac{EQL}) for different values of the pricing parameter $\lambda_0$ and step-size rules under the fast-fading regime.}
\label{fig:ergodic}
\end{figure}

\section{Conclusions}
\label{sec:conclusions}

In this paper, we considered a game-theoretic formulation of the problem of cost-efficient throughput maximization in \acl{MC} \ac{CR} networks where \acp{SU} are charged based on the interference that they cause to the system's \acp{PU}.
We showed that the resulting game admits a unique \acl{NE} under fairly mild conditions (and for both static and ergodic channels), and we derived a fully distributed learning algorithm that converges to equilibrium using only local \ac{SINR} and channel measurements (and, again, under both static and fast-fading channel conditions).
Our analysis shows that the choice of the exact pricing scheme has a strong impact on the network's achievable performance (for both licensed and unlicensed users):
in the ``soft-pricing'' regime, the \acp{PU}' requirements are violated in exchange for monetary reimbursement;
by contrast, higher prices safeguard the \acp{PU}' requirements, but (somewhat surprisingly) generate no revenue to the \acp{PU}.
Moreover, thanks to the fast convergence of the proposed algorithm, the system's transient (off-equilibrium) phase is minimized, so \acp{SU} avoid being unduly uncharged for relatively low throughput levels.

Some important questions that remain is the behavior of the system under arbitrarily time-varying channel conditions corresponding to more general fading models (not necessarily following a stationary ergodic process), and the case of imperfect \ac{SINR} measurements and channel knowledge at the transmitter.
We intend to explore these directions in future work.
\appendix[Technical Proofs]
\label{app:proofs}

\subsection{Equilibrium analysis}
\label{app:proofs-equilibrium}


\begin{IEEEproof}[Proof of Theorem \ref{thm:equilibrium}]
We will first show that the game's potential $\pot$ is strictly concave under assumption (A1) (i.e., if $\price_{k}$ is strictly increasing in each of its arguments).
To that end, let $\pot_{0} = \insum_{\sub} \log(\noisevar_{\sub} + w_{\sub}) - \price_{0}$, $\pot_{+} = -\insum_{k}\price_{k}$ and differentiate $\pot = \pot_{0} + \pot_{+}$ to obtain:
\begin{equation}
\label{eq:dV}
\frac{\pd\pot}{\pd p_{k\sub}}
	= \frac{\pd\pot_{0}}{\pd p_{k\sub}} + \frac{\pd\pot_{+}}{\pd p_{k\sub}}
	= \frac{\pd\pot_{0}}{\pd w_{\sub}} g_{k\sub}
	- \frac{\pd\price_{k}}{\pd p_{k\sub}},
\end{equation}
and hence:
\begin{equation}
\label{eq:HessV}
\frac{\pd^{2}\pot}{\pd p_{k\sub}\,\pd p_{\ell\sub'}}
	= g_{k\sub} g_{\ell\sub'} \frac{\pd^{2}\pot_{0}}{\pd w_{\sub}\,\pd w_{\sub'}}
	- \frac{\pd^{2}\price_{k}}{\pd p_{k\sub}\,\pd p_{k\sub'}} \delta_{k\ell}
	= - g_{k\sub} g_{\ell\sub'} A_{\sub\sub'}^{0}
	- \delta_{k\ell} B_{\sub\sub'}^{k},
\end{equation}
where, in obvious notation:
\begin{equation}
\label{eq:ABdef}
A_{\sub\sub'}^{0}
	= - \frac{\pd^{2}\pot_{0}}{\pd w_{\sub}\,\pd w_{\sub'}}
	\quad
	\text{and}
	\quad
B_{\sub\sub'}^{k}
	= \frac{\pd^{2}\price_{k}}{\pd p_{k\sub}\,\pd p_{k\sub'}}.
\end{equation}

Since $\pot_{0}$ is strictly concave in $\bw$ (as the sum of a strictly concave function and a concave function), it follows that $\{A_{\sub\sub'}^{0}\}$ is positive-definite.
Accordingly, since $A_{\sub\sub'}^{0}$ does not depend on $k$, any zero eigenvector $\bz\in\R^{KS}$ of the $KS\times KS$ matrix $g_{k\sub} g_{\ell\sub'} A_{\sub\sub'}^{0}$ must satisfy:
\begin{equation}
\label{eq:degeneracy}
\insum_{k} g_{k\sub} z_{k\sub}
	= 0
	\quad
	\text{for all $\sub\in\subs$}.
\end{equation}
The degeneracy condition \eqref{eq:degeneracy} reflects the fact that if $\bw(\bp') = \insum_{k} g_{k\sub} p_{k\sub}' = \insum_{k} g_{k\sub} p_{k\sub} = \bw(\bp)$ for two power profiles $\bp,\bp'\in\strat$, then $\pot_{0}(\bp) = \pot_{0}(\bp')$;
Eq.~\eqref{eq:degeneracy} shows in addition that $\pot_{0}$ admits no other directions along which it is constant.
From this, it follows that the kernel $Z$ of $\hess(\pot)$ is at most $S$-dimensional;
since $\argmax\pot$ lies in an affine subspace of $\R^{KS}$ that is parallel to $Z$, we conclude that the Nash set of $\game$ is a convex polytope of dimension at most $KS - S$, as claimed.

Assume now that $\eqvec$ is a Nash equilibrium of $\game$.
If there exists a subcarrier $\sub\in\subs$ such that $\eq_{k\sub} = 0$ for all $k\in\play$, then any profile with $p_{k\sub} = P_{k}$ for all $k\in\play$ 
cannot be Nash \textendash\ and vice versa.
Thus, without loss of generality (and after relabeling indices if necessary), we may assume that there exists a subcarrier $\sub\in\subs$ such that $\eq_{k\sub} < \eq_{\ell\sub}$ for two users $k,\ell\in\play$.
With this in mind, assume that every user-specific price function $\price_{k}$ is increasing in each of its arguments and consider the tangent
vector $\bz\in\R^{KS}$ with $z_{k\sub} = g_{\ell\sub}$, $z_{\ell\sub} = - g_{k\sub}$, and $z_{k'\sub'}=0$ otherwise.
By \eqref{eq:degeneracy}, it follows that
\begin{equation}
f(t)
	= \pot(\eqvec + t\bz)
\end{equation}
is constant for all sufficiently small $t\geq0$ (note that $\eqvec+t\bz\in\strat$ for small $t\geq0$).
However, by differentiating, we obtain:
\begin{equation}
\frac{df}{dt}
	= \frac{d}{dt} \left[ \pot_{0}(\eqvec + tz) - \txs\insum_{k'} \price_{k'}(\eqvec_{k'} + t\bz_{k'}) \right]
	= -\frac{\pd\price_{k}}{\pd p_{k\sub}} z_{k\sub} - \frac{\pd\price_{\ell}}{\pd p_{\ell\sub}} z_{\ell\sub}
	= g_{k\sub} \frac{\pd\price_{\ell}}{\pd p_{\ell\sub}} - g_{\ell\sub} \frac{\pd\price_{k}}{\pd p_{k\sub}},
\end{equation}
so we must have
\begin{equation}
g_{k\sub} \left. \frac{\pd \price_{\ell}}{\pd p_{\ell\sub}} \right\vert_{\eqvec+t\bz}
	= g_{\ell\sub} \left. \frac{\pd\price_{k}}{\pd p_{k\sub}} \right\vert_{\eqvec+t\bz}
	\quad
	\text{for all sufficiently small $t\geq0$.}
\end{equation}
With $\price_{k}$, $\price_{\ell}$ strictly increasing, this only holds if $\price_{k}$ (resp. $\price_{\ell}$) is linear in $p_{k\sub}$ (resp. $p_{k\sub}$) and the channel gain coefficients $g_{k\sub}$, $g_{\ell\sub}$ have the required ratio.
This last condition is a (Lebesgue) measure zero event, so our assertion follows.

Otherwise, assume that (A2) holds, implying in particular that $\frac{\pd\pot_{0}}{\pd w_{\sub}} = (\noisevar_{\sub} + w_{\sub})^{-1} - \frac{\pd\price_{0}}{\pd w_{\sub}}$ maintains the same sign for all possible values of $w_{\sub}$.
Then, in view of the previous discussion, it suffices to prove uniqueness in the special case where the price functions $\price_{k}$ are constant in a neighborhood of $\eqvec$.
In this case, the first order \ac{KKT} conditions for \eqref{eq:potential-max} take the form:
\begin{subequations}
\label{eq:KKT}
\begin{flalign}
a)\quad
	& r_{\sub} g_{k\sub} - \lambda_{k} \leq 0,
	\\
b)	\quad
	&p_{k\sub} \left[r_{\sub} g_{k\sub} - \lambda_{k} \right] = 0,
\end{flalign}
\end{subequations}
where $\lambda_{k}$ is the Lagrange multiplier corresponding to the total power constraint $\insum_{\sub} p_{k\sub} \leq P_{k}$ and
\begin{equation}
r_{\sub}
	= \left(\frac{1}{\noisevar_{\sub} + w_{\sub}} - \frac{\pd\price_{0}}{\pd w_{\sub}}\right)^{-1}.
\end{equation}
Thus, with $r_{\sub} \neq 0$ by assumption, we obtain:
\begin{equation}
\frac{g_{k\sub}}{g_{k\sub'}}
	= \frac{r_{\sub}}{r_{\sub}'}
	\quad
	\text{for all $\sub,\sub'\in\supp(\eqvec_{k})$},
\end{equation}
i.e., every user $k\in\play$ is ``load-balancing'' the quantity $g_{k\sub}/r_{\sub}$ over all employed subcarriers.

By using a graph-theoretic method introduced in \cite{MBML11}, we may deduce that the following hold except on a set of (Lebesgue) measure zero;
indeed:
\begin{enumerate}
\item
No two users $k,\ell\in\play$ can be using the same two subcarriers $\sub,\sub'$ at equilibrium:
if this were the case, we would have $g_{k\sub}/g_{k\sub'} = g_{\ell\sub}/g_{\ell\sub'}$, a measure zero event.
\item
There is at most $S-1$ instances of users employing more than one subcarrier.
Indeed, assume that user $k_{j}$ employs subcarriers $\sub_{j},\sub_{j}'$, with $j=1,\dotsc,N$, $N\geq S$.
Then, by the pigeonhole principle, there exists a subset of pairs $(\sub_{j},\sub_{j}')$ that forms a cycle of length $L\geq N$ in the graph with vertex set $\set$.
Hence, by relabeling indices if necessary, we obtain the cycle relation:
\begin{equation}
\frac{g_{k_{1},\sub_{1}}}{g_{k_{2},\sub_{2}}}
	\frac{g_{k_{2},\sub_{2}}}{g_{k_{3},\sub_{3}}}
	\dotsm
	\frac{g_{k_{L-1},\sub_{L-1}}}{g_{k_{L},\sub_{L}}}
	= \frac{r_{\sub_{1}}}{r_{\sub_{2}}}
	\frac{r_{\sub_{2}}}{r_{\sub_{3}}}
	\dotsm
	\frac{r_{\sub_{L-1}}}{r_{\sub_{L}}}
	= 1,
\end{equation}
where we have used the fact that $\sub_{1} = \sub_{L}$.
This represents a measure zero condition, so our assertion follows.
\end{enumerate}

The above shows that $\eqvec$ lies in the interior of a face of $\strat$ with dimension at most $S-1$.
Since the Nash set of $\game$ is a convex polytope of dimension $KS-S$, we conclude that any Nash equilibrium lies at the intersection of a $g$-independent $(S-1)$-dimensional and a $g$-dependent $(KS-S)$-dimensional subspace of $\R^{KS}$.
However, since $KS-S+S-1 < KS$, the intersection of these subspaces is trivial on a set of full (Lebesgue) measure with respect to the choice of the $g$-dependent subspace, implying that there exists a unique Nash equilibrium.
\end{IEEEproof}

\subsection{Convergence of exponential learning}
\label{app:proofs-convergence}

The basic idea of our convergence proof is as follows:
we will first show that the iterates of Algorithm \ref{alg:XL} track (in a certain sense that will be made precise below) the ``mean-field'' dynamics:
\begin{equation}
\label{eq:XL-cont}
\begin{aligned}
\dot \scorevec_{k}
	&= \payvec_{k}(\bp),
	\\
p_{ks}
	&= P_{k} \frac{\exp(y_{ks})}{1 + \insum_{s'\in\subs} \exp(y_{ks'})}.
\end{aligned}
\end{equation}
Theorem \ref{thm:conv} will then follow by showing that the dynamics \eqref{eq:XL-cont} converge to the maximum set of the game's potential (and, hence, to \acl{NE}) for any itial condition $\scorevec(0)$.

For simplicity, in the rest of this appendix (and unless explicitly stated otherwise), we will work with a single user with maximum transmit power $P=1$;
the general case is simply a matter of taking a direct sum over $k\in\play$ and rescaling by the corresponding maximum power $P_{k}$ of each user.
With this in mind, let $\varsimplex = \{\bp\in\R_{+}^{\subs}: 0 \leq \insum_{s} p_{s} \leq 1\}$ denote the standard $S$-dimensional ``corner-of-cube'',%
\footnote{Recall that each user's action space is a corner-of-cube.}
and consider the entropy-like function:
\begin{equation}
\label{eq:penalty}
h(\bp)
	= \insum_{s} p_{s} \log p_{s}
	+ \left(1 - \insum_{s} p_{s}\right) \log\left(1 - \insum_{s} p_{s}\right).
\end{equation}
A key element of our proof will be the associated \emph{Bregman divergence} \cite{Bre67,Kiw97}:
\begin{flalign}
\label{eq:Bregman}
\breg(\eqvec,\bp)
	= h(\eqvec) - h(\bp) - \braket{\nabla_{\bp}h}{\eqvec-\bp}
	= \insum_{s} \eq_{s} \log \frac{\eq_{s}}{p_{s}}
	+ \left(1 - \insum_{s} \eq_{s}\right) \log\frac{1 - \insum_{s} \eq_{s}}{1 - \insum_{s} p_{s}},
\end{flalign}
with the continuity convention $0\log0 = 0$.
The Bregman divergence \eqref{eq:Bregman} resembles the well known \ac{KL} divergence in the same sense that $h$ resembles the ordinary Gibbs\textendash Shannon entropy:
in particular, by exploiting the properties of the \ac{KL} divergence, it is easy to see that $\breg(\eqvec,\bp)\geq0$ for all $\eqvec,\bp\in\varsimplex$, with equality if and only if $\bp = \eqvec$;
in this sense, $\breg(\eqvec,\bp)$ provides an oriented distance measure between $\eqvec$ and $\bp$ in $\varsimplex$.

Employing the Bregman divergence, we can prove the following convergence result:

\begin{proposition}
\label{prop:conv-cont}
Every solution orbit $\bp(t)$ of the dynamics \eqref{eq:XL-cont} converges to Nash equilibrium in $\game$.
\end{proposition}
 
 \begin{proof}
 Let $\eqvec$ be a Nash equilibrium of $\game$,
and let $H(t) = \breg(\eqvec,\bp(t))$.
We then have:
\begin{equation}
\label{eq:H}
H
	= h(\eqvec)
	+ \log\left(1 + \insum_{s} e^{\score_{s}} \right)
	- \insum_{s} \eq_{s} \score_{s},
\end{equation}
and hence:
\begin{equation}
\label{eq:dH}
\dot H
	= \frac{\insum_{s} \dot\score_{s} e^{\score_{s}}}{1 + \insum_{s} e^{\score_{s}}}
	- \insum_{s} \eq_{s} \dot\score_{s}
	= \insum_{s} p_{s} \dot\score_{s}
	- \insum_{s} \eq_{s} \dot\score_{s}
	= \insum_{s} (p_{s} - \eq_{s}) \payv_{s}
	= \braket{\bp - \eqvec}{\payvec}.
\end{equation}
By concavity of $\pot$ and the fact that $\payvec = \grad_{\bp} \pay = \grad_{\bp}\pot$,
it follows that $\braket{\bp - \eqvec}{\payvec} \geq 0$ with equality holding if and only if $\bp$ is a maximizer of $\pot$ (and, hence, a Nash equilibrium of $\game$).

To show that $\bp(t)$ converges to a Nash equilibrium of $\game$, assume that $\eqvec$ is an $\omega$-limit of $\bp(t)$, i.e., $\bp(t_{n})\to\eqvec$ for some increasing sequence $t_{n}\to\infty$ (that $\bp(t)$ admits at least one $\omega$-limit follows from the fact that $\varsimplex$ is compact).
This implies that $H(t_{n}) \to0$, and since $\dot H \geq 0$, we also get $\lim_{t\to\infty} H(t) = 0$, so $\bp(t) \to \eqvec$ by the definition of the Bregman divergence.
\end{proof}

With this result at hand, we have:

\begin{proof}[Proof of Theorem \ref{thm:conv}]
We will first show that the basic recursion of Algorithm \ref{alg:XL} comprises a stochastic approximation of the dynamics \eqref{eq:XL-cont} in the sense of \cite{Ben99}.
Indeed, it is easy to see that the exponential regularization map \eqref{eq:Gibbs} is Lipschitz;
moreover, since $\varsimplex$ is compact and the game's potential function is smooth on $\varsimplex$, it follows that the composite map $\scorevec\mapsto\payvec(\bp(\scorevec))$ is also Lipschitz.
As a result, by Propositions 4.2 and 4.1 of \cite{Ben99}, we conclude that the recursion
\begin{equation}
\label{eq:XL}
\tag{XL}
\begin{aligned}
\scorevec(n+1)
	&= \scorevec(n) + \step_{n} \payvec(\bp(n)),
	\\
\bp(n+1)
	&= \frac{1}{1 + \sum_{s} e^{\score_{s}(n+1)}} (e^{\score_{1}(n+1)},\dotsc,e^{\score_{S}(n+1)}),
\end{aligned}
\end{equation}
is an \ac{APT} of the continuous-time dynamics \eqref{eq:XL-cont}.

Now, let $\eqset$ denote the set of Nash equilibria of $\game$, and assume ad absurdum that $\bp(n)$ remains a bounded distance away from $\eqset$.
Furthermore, fix some $\eqvec\in\eqset$ and let $D_{n} = \breg(\eq,\bp(n))$;
then, using \eqref{eq:dH}, we obtain the Taylor expansion:
\begin{flalign}
\label{eq:Dn}
D_{n+1}
	&= \breg(\eqvec,\bp(n+1))
	= \breg(\eqvec,\bp(\scorevec(n) + \step_{n} \payvec(\bp(n))))
	\notag\\
	&\leq D_{n}
	- \step_{n} \braket{\payvec(\bp(n))}{\eqvec - \bp(n)}
	+ \tfrac{1}{2} M \step_{n}^{2} \norm{\payvec(\bp(n))}^{2},
\end{flalign}
for some constant $M>0$ (that such a constant exists is a consequence of the fact that $\hess(h) \mgeq mI$ for some $m>0$ \cite{Nes09}).
Since $\bp(n)$ stays a bounded distance away from $\eqset$ (by assumption) and $\pot$ is concave, we will also have $\braket{\payvec(\bp(n))}{\eqvec - \bp(n)} \geq \delta$ for some $\delta>0$ and for all $n$.
Hence, telescoping \eqref{eq:Dn}, we get:
\begin{equation}
\label{eq:Dn1}
D_{n+1}
	\leq D_{0}
	- \delta \insum_{j=0}^{n} \step_{j}
	+ \frac{1}{2} M v^{2} \insum_{j=0}^{n} \step_{j}^{2},
\end{equation}
where we have set $v = \sup_{\bp\in\varsimplex} \norm{\payvec(\bp)}$.
Since $\sum_{j=1}^{\infty} \step_{j}^{2}\big/\sum_{j=1}^{\infty} \step_{j} \to 0$, this last inequality yields $\lim_{n\to\infty} D_{n} = -\infty$, a contradiction.
We thus conclude that $\bp(n)$ visits a compact neighborhood of $\eqset$ infinitely often, so our claim of convergence follows from \cite[Theorem 6.10]{Ben99}. 
\end{proof}

\subsection{The fast-fading case}
\label{app:proofs-fading}

Our goal in this appendix is to prove uniqueness of \ac{NE} in the ergodic game $\bar\game$ (Prop.~\ref{prop:potential-erg}) and the convergence of Algorithm \ref{alg:XL} in the presence of fast fading.

\begin{IEEEproof}[Proof of Proposition \ref{prop:potential-erg}]
That $\bar\pot$ is an exact potential for $\bar\game$ follows directly by inspection, as in the case of Proposition \ref{prop:potential}.
For the strict concavity of $\bar\pot$, let $\hess_{g}(\pot)$ denote the Hessian of the static potential function $\pot$ for a given realization of the channel gain coefficients $g$.
Then, with $\pot$ bounded and smooth over $\strat$, the dominated convergence theorem allows us to interchange differentiation and integration, so we obtain $\hess(\bar\pot) = \ex_{g}[\hess_{g}(\pot)]$.
Thus, for all $\bz\in\R^{KS}$, we will have:
\begin{equation}
\bz^{\dag} \cdot \hess(\bar\pot) \cdot \bz
	= \ex_{g} \big[ \bz^{\dag} \cdot \hess_{g}(\pot) \cdot \bz \big]
	\geq 0.
\end{equation}
From the proof of Theorem \ref{thm:equilibrium}, we know that $\bz^{\dag} \cdot \hess_{g}(\pot) \cdot \bz$ only if $\insum_{k} g_{k\sub} z_{k\sub} = 0$ for all $\sub\in\subs$;
however, since this is a measure zero event (recall that the law of $g$ is atom-free), we will have $\bz^{\dag} \cdot \hess_{g}(\pot) \cdot \bz > 0$ on a set of positive measure.
This shows that $\bz^{\dag} \cdot \hess(\bar\pot) \cdot \bz > 0$ for all $\bz\in\R^{KS}$, i.e., $\bar\pot$ is strictly concave.
We conclude that $\bar\game$ admits a unique equilibrium, as claimed.
\end{IEEEproof}

\begin{IEEEproof}[Proof of Theorem \ref{thm:conv-erg}]
The same reasoning as in the proof of Theorem \ref{thm:conv} shows that the iterates of Algorithm \ref{alg:XL} run with the players' instantaneous utilities calculated as in \eqref{eq:marginal-inst} comprise a stochastic approximation (\acl{APT}) of the mean dynamics:
\begin{equation}
\label{eq:XL-erg}
\begin{aligned}
\dot \scorevec_{k}
	&= \bar\payvec_{k}(\bp),
	\\
p_{ks}
	&= P_{k} \frac{\exp(y_{ks})}{1 + \insum_{s'\in\subs} \exp(y_{ks'})}.
\end{aligned}
\end{equation}
Again, by following the same steps as in the Proof of Theorem \ref{thm:conv}, we can show that the dynamics \eqref{eq:XL-erg} converge to the unique \acl{NE} of the ergodic game $\bar\game$;
as such, it suffices to show that any \ac{APT} of \eqref{eq:XL-erg} induced by Alg.~\ref{alg:XL} converges to equilibrium.

To that end, with notation as in \eqref{eq:Dn}, we readily obtain:
\begin{equation}
\label{eq:Dn-erg}
D_{n+1}
	= \breg(\eqvec,\bp(n+1))
	\leq D_{n}
	- \step_{n} \braket{\hat\payvec(n)}{\eqvec - \bp(n)}
	+ \tfrac{1}{2} M \step_{n}^{2} \norm{\hat\payvec(n)}^{2},
\end{equation}
where $\eqvec$ is the (unique) \ac{NE} of $\bar\game$ and $M>0$ is a positive constant.
Assume now that $\bp(n)$ remains a bounded distance away from $\eqvec$ (so $D_{n}$ is bounded away from zero),
and let $\xi_{n} = \braket{\hat\payvec(n) - \bar\payvec(\bp(n))}{\bp(n) - \eqvec}$.
Since $\bar\pot$ is (strictly) concave and $\bp(n)$ stays a bounded distance away from its maximum set, we will have $\braket{\bar\payvec(\bp(n))}{\eqvec - \bp(n)} \leq - m$ for some positive constant $m>0$.
Hence, telescoping \eqref{eq:Dn-erg} yields:
\begin{equation}
\label{eq:Dn-erg1}
D_{n+1}
	\leq D_{0}
	- t_{n} \left( m - \insum_{j=1}^{n} w_{j,n}\,\xi_{j} \right)
	+ \frac{1}{2} M \insum_{j=1}^{n} \step_{j}^{2} \norm{\hat \payvec(j)}^{2},
\end{equation}
where $t_{n} = \sum_{j=1}^{n} \step_{j}$ and $w_{j,n} = \step_{j}/t_{n}$.
By the strong law of large numbers for martingale differences \cite[Theorem 2.18]{HH80}, we will have $n^{-1} \sum_{j=1}^{n} \xi_{j} \to0$ (a.s.);
hence, with $\step_{n+1}/\step_{n}\leq1$, Hardy's weighted summability criterion \cite[p.~58]{Har49} applied to the weight sequence $w_{j,n} = \step_{j}/t_{n}$ yields $\sum_{j=1}^{n} w_{j,n}\, \xi_{j} \to 0$ (a.s.).
Finally, since $\step_{n}$ is square-summable and $\hat\payvec(n) - \bar\payvec(\bp(n))$ is a martingale difference with finite variance, it follows that $\sum_{n=1}^{\infty} \step_{n}^{2} \norm{\hat\payvec(n)}^{2} < \infty$ (a.s.) by Theorem 6 in \cite{Cho68}.

Combining all of the above, we obtain that the RHS of \eqref{eq:Dn1} tends to $-\infty$ (a.s.);
this contradicts the fact that $D_{n}\geq0$, so we conclude that $\bp(n)$ visits a compact neighborhood of $\eqvec$ infinitely often.
Since $\eqvec$ is a global attractor of \eqref{eq:XL-erg}, Theorem 6.10 in \cite{Ben99} shows that $\bp(n)$ converges to $\eqvec$ (a.s.).
\end{IEEEproof}


\balance
\bibliographystyle{IEEEtran}
\footnotesize
\bibliography{IEEEabrv,CognitivePricing}


%
%

\end{document}